\documentclass[envcountsame]{llncs}
\usepackage{pstricks,pst-node}
\usepackage{amsmath,amssymb}
\usepackage{enumerate}
\usepackage{pst-tree}
\usepackage{subfigure}
\usepackage{xspace,graphicx}

\usepackage{ifthen}

\providecommand{\version}{full} 

\newcounter{CommentCounter}
\newcommand{\acomment}[2]{\ \\ \fbox{\parbox{\linewidth}{{\sc #1}:\\ #2}}}

\newcommand{\ignore}[1]{}
\newcommand{\myops}[1]{\ensuremath{\mathrm{#1}}}
\newcommand{\mylog}[1]{\ensuremath{\mathrm{#1}}}
\newcommand{\myroot}{\ensuremath{\text{root}}}

\newenvironment{myproofof}[1]
{{\noindent\newline{\bf Proof of #1.}}}
{\qed}

\newcommand{\ahmet}[1]{\acomment{Ahmet}{#1}}

\newcommand{\HCTL}{\mylog{HCTL}\xspace}
\newcommand{\HCTLplus}{\mylog{HCTL^+}\xspace}
\newcommand{\HCTLstar}{\mylog{HCTL^\ast}\xspace}
\newcommand{\Hk}{\mylog{H^kCTL}\xspace}

\newcommand{\Hone}{\mylog{H^1CTL}\xspace}

\newcommand{\HkE}{\mylog{H^kECTL}\xspace}
\newcommand{\Honeplus}{\mylog{H^1CTL^+}\xspace}
\newcommand{\Hkplus}{\mylog{H^kCTL^+}\xspace}
\newcommand{\Hkstar}{\mylog{H^kCTL^*}\xspace}

\newcommand{\Honestar}{\mylog{H^1CTL^\ast}\xspace}
\newcommand{\CTL}{\mylog{CTL}\xspace}
\newcommand{\CTLplus}{\mylog{CTL^+}\xspace}
\newcommand{\CTLstar}{\mylog{CTL^\ast}\xspace}
\newcommand{\ECTL}{\mylog{ECTL}\xspace}
\newcommand{\HoneplusPastFair}{\mylog{H^1PECTL^+}\xspace}
\newcommand{\HonePast}{\mylog{H^1PCTL}\xspace}

\newcommand {\calA}      {{\cal A}\xspace}
\newcommand {\calB}      {{\cal B}\xspace}

\newcommand {\calF}      {{\cal F}\xspace}
\newcommand {\calT}      {{\cal T}\xspace}
\newcommand {\calU}      {{\cal U}\xspace}
\newcommand {\calL}      {{\cal L}\xspace}

\newcommand{\E}{\myops{E}}
\newcommand{\A}{\myops{A}}
\newcommand{\X}{\myops{X}}
\newcommand{\F}{\myops{F}}
\newcommand{\G}{\myops{G}}
\newcommand{\U}{\myops{U}}
\newcommand{\Y}{\myops{Y}}
\newcommand{\Since}{\myops{S}}
\newcommand{\dx}{\ensuremath{{\downarrow}x}}

\newcommand{\ax}[1][]{\ensuremath{@_{x_{#1}}}}
\newcommand{\ar}{\ensuremath{@_\myroot}}
\newcommand{\Finfty}{\ensuremath\overset{\infty}{\F}}
\newcommand{\Ginfty}{\ensuremath\overset{\infty}{\G}}

\newcommand{\rowo}{\ensuremath{\mathit{row}_o}}
\newcommand{\rowe}{\ensuremath{\mathit{row}_e}}
\newcommand{\poso}{\ensuremath{\mathit{pos}_o}}
\newcommand{\pose}{\ensuremath{\mathit{pos}_e}}
\newcommand{\pos}{\ensuremath{\varphi_{\mathit{pos}}}}
\newcommand{\row}{\ensuremath{\varphi_{\mathit{row}}}}
\newcommand{\psicur}{\ensuremath{\psi_{\mathit{full}}}}
\newcommand{\phifirst}{\ensuremath{\varphi_{\mathit{first}}}}
\newcommand{\philast}{\ensuremath{\varphi_{\mathit{last}}}}
\newcommand{\psipos}{\ensuremath{\psi_{\mathit{2pos}}}}
\newcommand{\psirow}{\ensuremath{\psi_{\mathit2row}}}

\newcommand{\size}[1]{\ensuremath{|#1|}}
\newcommand{\bigO}{\ensuremath{\mathcal{O}}}
\newcommand{\Nat}{\ensuremath{\mathbb{N}}}

\newcommand  {\myclass} [1]  {{\ensuremath{\mbox{\bf #1}}}\xspace}

\newcommand     {\Dexptime}  {\myclass{2EXPTIME}}
\newcommand     {\Texptime}  {\myclass{3EXPTIME}}

\newcommand{\dashednode}[2]{\cnode[linestyle=dashed](#1){.2}{#2}}
\newcommand{\wnode}[2]{\cnode(#1){.2}{#2}}
\newcommand{\bnode}[2]{\cnode[fillstyle=solid,fillcolor=black](#1){.2}{#2}}
\newcommand{\gnode}[2]{\cnode[fillstyle=solid,fillcolor=gray](#1){.2}{#2}}

\newcommand{\fullonly}[1]{\ifthenelse{\equal{\version}{full}}{#1}{}}
\newcommand{\mfcsonly}[1]{\ifthenelse{\equal{\version}{mfcs}}{#1}{}}
\newcommand{\mfcsorfull}[2]{\ifthenelse{\equal{\version}{mfcs}}{#1}{#2}}

\newcommand{\ourcondition}[2]{#1: \emph{#2}}
\newcommand{\blow}{{\mathtt{l}}}
\newcommand{\bflip}{{\mathtt{f}}}
\newcommand{\bstay}{{\mathtt{s}}}

\fullonly{
\addtolength{\textheight}{2.7cm}
\addtolength{\textwidth}{2.1cm}
\addtolength{\oddsidemargin}{-1.2cm}
\addtolength{\evensidemargin}{-1.2cm}
\addtolength{\topmargin}{-1.3cm}
\addtolength{\intextsep}{-5mm}

}

\author{Ahmet Kara$^1$ \and Volker Weber$^1$ \and Martin Lange$^2$ \and Thomas Schwentick$^1$ }
\institute{Technische Universit\"{a}t Dortmund \and
  Ludwig-Maximilians-Universit\"{a}t M\"{u}nchen} 

\title{On the Hybrid Extension of \CTL and \CTLplus} 

\begin{document}
\maketitle


\begin{abstract}
  The paper studies the expressivity, relative succinctness and complexity of satisfiability for hybrid
  extensions of the branching-time logics \CTL and \CTLplus by variables. Previous complexity results
  show that only fragments with \emph{one variable} do have elementary complexity. It is shown that
  \Honeplus and \Hone, the hybrid extensions with
  one variable of \CTLplus and \CTL, respectively, are expressively equivalent but \Honeplus is exponentially more succinct than
  \Hone. On the other hand, \HCTLplus, the hybrid extension of \CTL
  with arbitrarily many variables does not capture \CTLstar,  as it
  even cannot express the simple \CTLstar property $\E\G\F p$. The satisfiability problem for
  \Honeplus is complete for triply exponential time, this remains true for quite weak fragments and
  quite strong extensions of the logic.
\end{abstract}

\section{Introduction}

\newcommand{\mysl}[1]{{\blue #1}\xspace}

Reasoning about trees is at the heart of many fields in computer
science
\fullonly{, such as verification and semistructured data}.
 A wealth of sometimes quite different frameworks has been proposed for this
purpose, according to the needs of the respective application. For
reasoning about computation trees as they occur in verification,
branching-time logics like \CTL and tree automata are two such
frameworks. 
\fullonly{In fact, they are closely related \cite{Vardi95}. }
\fullonly{\\}
In some settings, the ability to mark a node in a tree and to refer to
this node turned out to be useful. As neither classical branching-time
logics nor tree automata provide this feature, many different
variations have been considered, including tree automata with pebbles
\cite{EngelfrietH99,CateS08,Weber09},  memoryful \CTLstar
\cite{KupfermanV06}, branching-time logics with forgettable past
\cite{LaroussinieS95,LaroussinieS00}, and logics with the ``freeze''
operator \cite{JurdzinskiL07}%
\fullonly{, the latter ones in the context of data trees \cite{Segoufin06}}.
\mfcsorfull{
It is
an obvious question how this feature can be incorporated into branching-time logics 
{\em without losing their desirable properties} which made them prevailing in
verification \cite{Vardi08}.
}
{
As classical logic naturally provides means to refer to a node, namely
constants and variables, it is
an obvious question how these means can be incorporated into branching-time logics 
without losing their desirable properties which made them prevailing in
verification \cite{Vardi08}.}

 This question leads into the field of hybrid logics, where such extensions
of temporal logics are studied \cite{ArecesC07}. In particular, a hybrid extension of \CTL has been
introduced in \cite{Weber09}.
\fullonly{\\}
 As usual for branching-time logics, formulas of their hybrid extensions are evaluated at nodes of a computation
tree, but it is possible to bind a variable to the current node, to
evaluate formulas relative to the root and to check whether the
current node is bound to a variable. As an example, the \HCTL-formula $\dx \ar \E\F (p \land \E\F x)$ intuitively says
``I  can place $x$ at the current node, jump back to the root, go to a
node where $p$ holds and follow some (downward) path to reach $x$
again´´. Or, equivalently: ``there was a node fulfilling
$p$ in the past of the current node''.    

In this paper we continue the investigation of hybrid extensions of classical branching-time logics
started in \cite{Weber09}.  The main questions considered are (1)~expressivity, (2)~complexity of the
satisfiability problem, and (3)~succinctness. Figure \ref{fig:results} shows our results in their
context. 
\fullonly{
The complexity of the model checking problems will be studied in future work.
}

\definecolor{mycolor1}{rgb}{1.0,1.0,0.6}
\definecolor{mycolor2}{rgb}{0.8,0.9,0.6}
\definecolor{mycolor3}{rgb}{0.9,0.8,1.0}
\definecolor{mycolor4}{rgb}{1.0,0.6,0.6}
\definecolor{mycolor5}{rgb}{0.9,1.0,1.0}
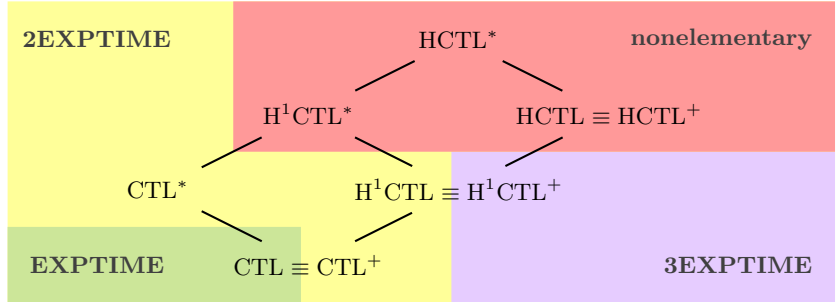
\begin{figure}[t]\label{fig:results}
\begin{center}
\mfcsonly{\psset{unit=7mm}\footnotesize}
\begin{pspicture}(-1,0.3)(10,4)
\psframe[fillstyle=solid,fillcolor=mycolor1,linecolor=mycolor1](-1,0)(10,4)
\psframe[fillstyle=solid,fillcolor=mycolor2,linecolor=mycolor2](-1,0)(2.9,1)
\psframe[fillstyle=solid,fillcolor=mycolor3,linecolor=mycolor3](4.9,0)(10,2)
\psframe[fillstyle=solid,fillcolor=mycolor4,linecolor=mycolor4](2,2)(10,4)

\color{darkgray}
\rput[l](-0.7,0.5){\textbf{EXPTIME}}
\rput[l](-0.8,3.5){\textbf{2EXPTIME}}
\rput[r](9.7,0.5){\textbf{3EXPTIME}}
\rput[r](9.7,3.5){\textbf{nonelementary}}
\color{black}

\rput(3,0.5){\rnode{node1}{$\CTL\equiv\CTLplus$}}

\rput(1,1.5){\rnode{node2}{\CTLstar}}

\rput(5,1.5){\rnode{node3}{$\Hone\equiv\Honeplus$}}

\rput(3,2.5){\rnode{node4}{\Honestar}}

\rput(7,2.5){\rnode{node5}{$\HCTL\equiv\HCTLplus$}}

\rput(5,3.5){\rnode{node6}{\HCTLstar}}

\psset{nodesep=5pt}
\ncline{node1}{node2}
\ncline{node1}{node3}
\ncline{node2}{node4}
\ncline{node3}{node4}
\ncline{node3}{node5}
\ncline{node4}{node6}
\ncline{node5}{node6}

\end{pspicture}
\end{center}
\label{fig:intro}
\caption{Expressivity and complexity of satisfiability for hybrid branching-time logics. The lines indicate strict inclusion, unrelated logics are incomparable.}
\end{figure}

Classical branching-time logics are \CTL (with polynomial time model checking and
exponential time satisfiability) and \CTLstar (with polynomial space model checking and
doubly exponential time satisfiability test). As \CTL is sometimes not
expressive enough\footnote{Some things cannot be expressed at all, some only in a very verbose way.} and \CTLstar is
considered too expensive for some applications, there has been an
intense investigation of intermediate logics. We take up two of them
here: \CTLplus, where a path formula is a Boolean combination of
basic path formulas\footnote{Precise definitions can be found in Section
  \ref{sec:definitions}.} and \ECTL, where fairness properties can be
stated explicitly. 

Whereas (even simpler) hybrid logics
are undecidable  over arbitrary transition systems \cite{ArecesBM00},
their restriction to trees is decidable via a simple translation to
Monadic Second Order logic. However, the complexity of the
satisfiability problem is high even for simple hybrid temporal logics
over the frame of natural numbers: nonelementary
\mfcsorfull{\cite{FranceschetRS03short}}{\cite{FranceschetRS03}}
, even if only two variables are allowed 
\mfcsorfull{\cite{SchwentickW07short,Weber09}}{\cite{SchwentickW07,Weber09}}. The one variable extension of \CTL, $\Hone$, behaves considerably
better, its satisfiability problem can be solved in
\Dexptime~\cite{Weber09}. This is the reason why this paper
concentrates on natural extensions of this complexity-wise relatively modest logic.
\fullonly{\\}
Even \Hone can express properties that are not bisimulation-invariant
(e.g., that a certain configuration can be reached along two distinct
computation paths) and is thus not captured by
\CTLstar. In fact, \cite{Weber09}  shows that \Hone captures and is strictly
stronger than \CTL with past, another extension of \CTL studied in
previous work \cite{KupfermanP95}. One of our main results is that $\Hone$ (and
actually even \HCTLplus) does not capture \ECTL (and therefore not \CTLstar) as it cannot
express simple fairness properties like $\E\G\F p$. 
To this end, we introduce a simple Ehrenfeucht-style game  (in the
spirit of \cite{ArecesBM01}).
We show that existence of a winning strategy for the
second player in the game for a property $P$ implies that $P$ cannot
be expressed in \HCTLplus.

In \cite{Weber09} it is also shown that the
satisfiability problem for \Honestar has nonelementary complexity. We
show here that the huge complexity gap between \Hone and \Honestar  does not yet occur between
\Hone and \Honeplus:  we prove that there is only an
exponential complexity gap between \Hone and \Honeplus, even when
\Honeplus is extended by past modalities and
fairness operators. We pinpoint the exact complexity by proving the problem
complete for \Texptime.

The exponential gap between the complexities for satisfiability of
\Hone and \Honeplus already suggests that \Honeplus might be
exponentially more succinct than \Hone. In fact, we show an exponential
succinctness gap between the two logics by a proof based on the height
of finite models. 
\fullonly{This refines the method of \cite{ipl-ctlplus08}
based on model size.}
It should be noted that an $\bigO(n)!$-succinctness gap between \CTL
and \Hone was established in \cite{Weber09}. We mention that there are
other papers on hybrid logics and hybrid tree logics that do not
study expressiveness or complexity issue, e.g.,
\cite{Goranko00,SattlerV01}.

The paper is organized as follows. Definitions of the logics we use
are in Section
\ref{sec:definitions}. Expressivity results are presented in Section
\ref{sec:expressivity}. The complexity results can be found in Section
\ref{sec:complexity}, the succinctness results in Section
\ref{sec:succinctness}. 
\mfcsonly{Proofs omitted due to space constraints can be
found in the full version of this paper \cite{KaraLSW09}.}


\section{Definitions}\label{sec:definitions}

\subsubsection*{Tree logics.}

\mfcsorfull{
We first  review the definition of \CTL and \CTLstar \cite{ClarkeE81}. 
}
{In this section, we define syntax and semantics of the logics we
use. 
We assume the reader is familiar with the tree logics \CTL and
\CTLstar \cite{ClarkeE81}. 
However, we review the definition of the syntax and semantics of them
next.}
Formulas of \CTLstar are composed from \emph{state formulas} $\varphi$ and \emph{path
  formulas} $\psi$. They have the following abstract syntax.
\begin{align*}                  
\varphi::= &\ p \mid\neg \varphi \mid\varphi \lor \varphi \mid\varphi \land \varphi \mid\E \psi \mid\A \psi  \\
\psi::= &\ \varphi \mid\neg \psi \mid\psi \lor \psi \mid\psi \land \psi \mid\X\psi  \mid\psi \U \psi    
\end{align*}

We use the customary abbreviations $\F\psi$ for $\top \U \psi$ and
$\G\psi$ for $\neg\F\neg\psi$. 
\fullonly{\\}
The semantics of formulas is defined inductively. 
The semantics of path formulas is defined relative to a
tree\footnote{In general, we consider finite and infinite trees and,
  correspondingly, finite and infinite paths in trees. It should
  always be clear from the context whether we restrict attention to
  finite or infinite trees.}
 $\calT$,
a path $\pi$ of $\calT$ and a position $i\ge 0$ of this path. E.g.,
$\calT,\pi,i\models \psi_1\U\psi_2$ if there is some $j\ge i$ such
that $\calT,\pi,j\models \psi_2$ and, for each $l, i\le l<j$,
$\calT,\pi,l\models\psi_1$. 
\fullonly{\\}
The semantics of state formulas is defined relative to a tree $\calT$
and a node $v$ of $\calT$. E.g., $\calT,v\models \E\psi$ if there
is a path $\pi$ in $\calT$, starting from $v$ such that
$\calT,\pi,0\models\psi$. A state formula $\varphi$ holds in a tree
$\calT$ if it holds in its root. Thus, sets of trees can be defined by
\CTLstar state formulas.

 \CTL is a strict sub-logic of \CTLstar. It allows only path
 formulas of the forms $\X\varphi$ and $\varphi_1\U\varphi_2$ where
 $\varphi,\varphi_1,\varphi_2$  are state formulas.  
\CTLplus is the sub-logic of \CTLstar where path formulas are Boolean
combinations of formulas of the forms  $\X\varphi$ and $\varphi_1\U\varphi_2$
and $\varphi,\varphi_1,\varphi_2$ are state formulas.

\ignore{
We use the abbreviation $\Finfty\psi$ for the path formula $\G\F\psi$, expressing that the path
has infinitely many nodes in which $\psi$ holds.
}
\subsubsection*{Hybrid logics.} In hybrid logics, a limited use of
variables is allowed. For a general introduction to hybrid logics we
refer to \cite{ArecesC07}. As mentioned in the introduction,
we concentrate in this paper on hybrid logic formulas with \emph{one}
variable $x$. However, as we also discuss logics with more variables,
we define hybrid logics $\Hkstar$ with $k$ variables. For each $k\ge 1$, the syntax of \Hkstar is defined by extending \CTLstar
with the following rules for state formulas.
\begin{align*}
\varphi::=\dx_i\,\varphi \mid x_i \mid @_{x_i}\,\varphi \mid \myroot \mid @_\myroot\, \varphi
\end{align*}
where $i\in\{1,\ldots,k\}$.
The semantics is now relative to a vector $\vec
u=(u_1,\ldots,u_k)$ of nodes of $\calT$ representing an
assignment $x_i\mapsto u_i$. For a node $v$ and $i\le k$ we write
$\vec u[i/v]$ to denote $(u_1,\ldots,u_{i-1},v,u_{i+1},\ldots,u_k)$.
For a tree $\calT$ a node $v$ and a
vector $\vec u$, the semantics of the new state formulas is defined as
follows. 
\begin{displaymath}
\begin{array}{lcl}
\calT,v,\vec u\models \dx_i\,\varphi &\enspace\mbox{if}\enspace& \calT,v,\vec u[i/v]\models\varphi \\
\calT,v,\vec u\models x_i &\mbox{if}&  v=u_i \\
\calT,v,\vec u\models \ax[i]\, \varphi &\mbox{if}& \calT,u_i,\vec u\models\varphi \\
\calT,v,\vec u\models \myroot &\mbox{if}& v \mbox{ is the root of } \calT \\
\calT,v,\vec u\models @_\myroot\, \varphi &\mbox{if}& \calT,r,\vec u\models\varphi , 
  \mbox{ where } r \mbox{ is the root of } \calT
\end{array}
\end{displaymath}
Similarly, the semantics of path formulas is defined relative to a
tree $\calT$, a path $\pi$ of $\calT$, a position $i \geq 0$ of $\pi$
and a vector $\vec u$. 
\fullonly{E.g.,
$\calT,\pi,i, \vec u \models \X \psi$ if $\calT,\pi,i+1, \vec u \models \psi$. 
}
\fullonly{\\}
Intuitively, to evaluate a formula $ \dx_i\,\varphi$ one puts a pebble
$x_i$ on the current node $v$ and evaluates $\varphi$. During the
evaluation, $x_i$ refers to $v$ (unless it is bound again by another
$\dx_i$-quantifier).   

The hybrid logics \Hkplus and \Hk are obtained by restricting 
\Hkstar in the same fashion as for \CTLplus and \CTL,
respectively. The logic $\HCTL$ is the union of all logics $\Hk$,
likewise \HCTLplus and \HCTLstar.

\ignore{
We also consider the extensions of \CTL and \Hk by formulas of the types
$\E\Finfty\varphi$ and $\E\Ginfty\varphi$. The semantics is defined as follows:
\begin{align*}
\calT,v,\vec u \models \E\Finfty\varphi \enspace &\mbox{if there is a path } \pi
  \mbox{ starting from } v \mbox{ which has infinitely many} \\
  &\mbox{nodes } v' \mbox{ with } \calT,v',\vec u\models \varphi \\
\calT,v,\vec u \models \E\Ginfty\varphi \enspace &\mbox{if there is a path } \pi
  \mbox{ starting from } v \mbox{ such that all but finitely} \\
  &\mbox{many nodes } v' \mbox{ of } \pi \mbox{ fulfill } \calT,v', \vec u\models \varphi 
\end{align*}
The resulting logics are denoted by \ECTL and \HkE, respectively.
}

\mfcsorfull{(Finite) satisfiability of formulas, the notion of a model
and equivalence of two (path and state) formulas $\psi$ and
  $\psi'$ (denoted $\psi \equiv \psi'$) are defined in the obvious
  way.}
{
A state formula $\varphi$ of a hybrid logic is \emph{satisfiable} if there exists a tree $\calT$ with $\calT, r, \vec u \models \varphi$, where $r$ is the root of $\calT$ and $\vec u = (r,\ldots, r)$ is a vector of adequate length. In this case we also say that $\calT$ is a \emph{model} of $\varphi$ (denoted as $\calT \models \varphi$). A state formula $\varphi$ is \emph{finitely satisfiable} if it has a \emph{finite} model. 

Two path formulas $\psi$ and $\psi'$ are \emph{equivalent} (denoted as
$\psi \equiv \psi'$) if for all trees $\calT$, all paths $\pi$ of
$\calT$ and all vectors $\vec u$ of adequate length it holds: $\calT,\pi,0,\vec u \models
\psi$ iff $\calT,\pi,0,\vec u \models \psi'$. Similarly, two state formulas
$\varphi$ and $\varphi'$ are equivalent (denoted as $\varphi \equiv
\varphi'$) if for all trees $\calT$, all nodes $v$ and all vectors $\vec u$ of adequate length it holds:
$\calT,v,\vec u\models\varphi$ iff $\calT,v,\vec u\models\varphi'$.
}
We say that a logic $\calL'$ is at least as expressive as
$\calL$ (denoted as $\calL \leq \calL'$) if for every
$\varphi \in \calL$ there is a $\varphi' \in \calL'$ such that
$\varphi \equiv \varphi'$.
$\calL$ and $\calL'$ have the 
\emph{same expressive power} if $\calL \leq \calL'$ and
$\calL' \leq \calL$. $\calL'$ is \emph{strict more expressive} than $\calL$ 
if $\calL \leq \calL'$ but not $\calL' \leq \calL$.  

\subsubsection*{Size, depth and succinctness.}
For each formula $\varphi$, we define its \emph{size}
$\size{\varphi}$ as usual and its \emph{depth} $d(\varphi)$ as the
nesting depth with respect to path quantifiers.

\fullonly{
It should be remarked that the definition of $d(\varphi)$ is tailored
for the proof of inexpressibility with respect to \Hk. For general
\Hkstar formulas one would count also the nesting of temporal operators.
}

The formal notion of \emph{succinctness} is a bit delicate. We follow the
approach of \cite{GroheS05} and refer to the discussion there.
We say that a logic $\calL$ is \emph{$h$-succinct in} a logic $\calL'$, for a function
$h: \mathbb{N} \rightarrow \mathbb{R}$, if for every formula $\varphi$
in $\calL$ there is an equivalent formula $\varphi'$ in $\calL'$ such
that $\size{\varphi'}\le h(\size{\varphi})$. $\calL$ is
\emph{$\calF$-succinct in} $\calL'$ if  $\calL$ is
$h$-succinct in $\calL'$, for some $h$ in function class $\calF$.  We say that $\calL$ is
\emph{exponentially more succinct} than $\calL'$ if $\calL$ is \emph{not}
$h$-succinct in $\calL'$, for any function $h\in 2^{o(n)}$. 

\subsubsection*{Normal forms.}
\mfcsorfull
{
We say that a \Hk formula is in
\emph{\E-normal form}, if it does not use the path quantifier \A \ at
all. A formula is in \emph{\U-normal form} if it only uses the
combinations $\E\X$, $\E\U$ and $\A\U$ (but not, e.g., $\E\G$ and
$\A\X$).} 
{
It will sometimes be convenient to restrict the set of operators that
have to be considered. To this end, we say that a \Hk formula is in
\emph{\E-normal form}, if it does not use the path quantifier \A \ at
all. A formula is in \emph{\U-normal form} if it only uses the
combinations $\E\X$, $\E\U$ and $\A\U$ (but not, e.g., $\E\G$ and
$\A\X$).} 
\mfcsorfull{
\begin{proposition}
Let $k\ge 1$. For each \Hk formula $\varphi$ there is an equivalent
    \Hk-formula of linear size in $\U$-normal form and an equivalent
    \Hk-formula in $\E$-normal form.
\end{proposition}
}
{
\begin{proposition}
Let $k\ge 1$. 
  \begin{enumerate}[(a)]
  \item For each \Hk formula $\varphi$ there is an equivalent
    \Hk-formula in $\U$-normal form and the size of $\psi$ is linear
    in the size of $\varphi$.
  \item For each \Hk formula $\varphi$ there is an equivalent
    \Hk-formula in $\E$-normal form.
  \end{enumerate}
\end{proposition}
\begin{proof}
  \begin{enumerate}[(a)]
  \item This can be easily shown just as for \CTL. Actually, the
    original definition of  \CTL
    by Emerson and Clarke \cite{ClarkeE81} used only $\E\X$, $\E\U$ and $\A\U$.
  \item This is straightforward as $\A(\psi\U\chi)$ can equivalently expressed as 
        $(\neg\E(\neg\chi)$ $\U(\neg\chi\land\neg\psi)) \land (\neg\E\G\neg\chi)$. 
However, it should be noted
    that the recursive application of this replacement rule may result in a
    formula of exponential size.
  \end{enumerate}
\qed
\end{proof}
}


\section{Expressivity of \HCTL and \HCTLplus}\label{sec:expressivity}
\subsection{The expressive power of \HCTLplus compared to \HCTL}
Syntactically \CTLplus extends \CTL by allowing Boolean combinations of path formulas in the scope of a
path quantifier $\A$ or $\E$. Semantically this gives \CTLplus the ability to fix a path and test its
properties by \emph{several} path formulas. However in 
\mfcsorfull{\cite{EmersonH82short}}{\cite{EmersonH82}}
it is shown that every
\CTLplus-formula can be translated to an equivalent \CTL-formula. The techniques used there are
applicable to the hybrid versions of these logics.
\begin{theorem} \label{sameExpPower}
For every $k \geq 1$, $\Hk$ has the same expressive power as $\Hkplus$.
\end{theorem}
\mfcsorfull{
\begin{proof}[Sketch]
  For a given $k \geq 1$ it is clear that every $\Hk$-formula is also a $\Hkplus$-formula. It remains
  to show that every $\Hkplus$-formula can be transformed into an equivalent $\Hk$-formula.  In
\mfcsorfull{\cite{EmersonH82short}}{\cite{EmersonH82}}
, rules for the transformation of a \CTLplus formula into an equivalent \CTL formula
  are given. Here, we have to consider the additional case in which a subformula in the scope of the
  \dx-operator is transformed. However, it is not hard to see that the transformation extends to this
  case as any assignment to a variable $x$ can be viewed as a proposition that only holds in one node.
  It should be noted that for a \Hkplus-formula $\varphi$ the whole transformation constructs a
  \Hk-formula of size $2^{\bigO(\size{\varphi} \log \size{\varphi})}$. 
\qed
\end{proof}
} 
{ 
\begin{proof}
  The main difficulty in the translation from \CTLplus to \CTL can be described as follows: In a
  formula like $\E[\F\varphi_1 \wedge \ldots \wedge \F\varphi_n]$ it is not determined in which order the
  formulas $\varphi_1,\ldots,\varphi_n$ hold on the path fixed by the quantifier $\E$. In
  \cite{EmersonH82} this problem is solved by listing \emph{all} possible orders. For instance the
  formula $\varphi=\E[\F\varphi_1\land\F\varphi_2]$ is equivalent to $\varphi' = \E\F(\varphi_1 \wedge
  \E\F\varphi_2) \vee \E\F(\varphi_2 \wedge \E\F\varphi_1)$. The transformation algorithm for \CTLplus
  to \CTL in \cite{EmersonH82} is based on the following equivalences of \CTLplus-formulas\footnote{In
    \cite{EmersonH82} and its journal version \cite{EmersonH85} there are some slight inaccuracies in
    the equivalences. Here we list the corrected ones.}:
{\small
\begin{enumerate}[(1)]
\item $\neg \X \varphi \equiv \X \neg \varphi$
\item $\neg (\varphi \U \varphi') \equiv [(\varphi \wedge \neg \varphi') \U (\neg \varphi \wedge \neg \varphi')] \vee \G \neg \varphi'$
\item $\E(\psi \vee \psi') \equiv \E \psi \vee \E \psi'$
\item $\X\varphi \wedge \X \varphi' \equiv \X(\varphi \wedge \varphi')$
\item $\G\varphi \wedge \G \varphi' \equiv \G(\varphi \wedge \varphi')$
\item 
$\displaystyle\E[\bigwedge_{i=1}^{n} (\varphi_i \U \varphi'_i)\wedge \X \chi \wedge
\G \xi ] \equiv$ 
\mfcsonly{\\\mbox{}\hfill}
$\displaystyle\bigvee_{I \subseteq \{1,...,n\}}[\bigwedge_{i \in I}\varphi'_i \wedge \xi \wedge \bigwedge_{i \notin I} \varphi_i \wedge \E\X(\chi \wedge \E(\bigwedge_{i \notin I} (\varphi_i \U \varphi'_i)\wedge \G \xi))]$\\
\item
$\displaystyle\E[\bigwedge_{i=1}^n (\varphi_i \U \varphi'_i) \wedge \G
\chi] \equiv \bigvee_{\pi \in Permutation(\{1,...,n\})}[\E[(\bigwedge_{i=1}^n \varphi_i \wedge \chi)\U(\varphi'_{\pi(1)} \wedge$\\\mbox{}\hfill
$\displaystyle \E[(\bigwedge_{i \not= \pi(1)}\varphi_i \wedge
\chi)\U(\varphi'_{\pi(2)}\wedge \E[(\bigwedge_{i \not= \pi(1),\pi(2)}\varphi_i \wedge \chi)\U(\varphi'_{\pi(3)}\ldots\U(\varphi'_{\pi(n)}\wedge \E\G \chi)\ldots)])])]]$\\
\end{enumerate} 
}
As explained above in equivalence (7) a disjunction of all possible orders of the formulas $\varphi'_i$ is
formulated.  These equivalences also hold for \Hkplus. It can easily be shown that the occurence of $\myroot$, $\ar$, $\dx$ or
 $\ax$ for a variable $x$ does not destroy any of the equivalences. Furthermore, as
already indicated, if a formula on the right or the left side of one of the equivalences is in the
scope of $\dx$ then the node to which $x$ is assigned is (up to a new \dx) unique which means that $x$
can be treated like a usual proposition. Altogether the translation algorithm for \CTLplus to \CTL
presented in \cite{EmersonH82} also gives a translation algorithm from \Hkplus to \Hk for every $k \geq 1$.
In \cite{EmersonH82} it is noticed that the factorial blowup introduced by equivalence (7) is the worst blowup
in the whole transformation process and since $n! =  2^{\bigO (n \ log\ n)}$ the transformation of a formula $\varphi$ 
results in a formula of length $2^{\bigO (n \ log\ n)}$.
\qed
\end{proof}
}
      
The transformation algorithm in Theorem \ref{sameExpPower} also yields an upper bound for the succinctness between \Honeplus and \Hone.  
\begin{corollary}\label{succUpperBound}
\Honeplus is $2^{\bigO(n \log n)}$-succinct in \Hone.
\end{corollary}

\newcommand{\nvec}{\ensuremath{\vec u}}
\newcommand{\com}{\ensuremath{u}}

\subsection{Fairness is not expressible in \HCTLplus.}

In this subsection, we show the following result.

\begin{theorem}\label{theo:fairness}
There is no formula in \HCTLplus which is logically  equivalent to $\E\Finfty p$.
\end{theorem}
Here, $\calT,v,\vec u \models \E\Finfty\varphi$ if there is a path $\pi$
starting from $v$ that has infinitely many nodes $v'$  with
$\calT,v',\vec u\models \varphi$.  
As an immediate consequence of this theorem, \HCTLplus  does not capture \CTLstar. 

In order to prove Theorem \ref{theo:fairness}, we  define an Ehrenfeucht-style game that corresponds to the expressive power of
\HCTL. A game for a different hybrid
logic was studied in \cite{ArecesBM01}. We show that if a set $L$ of trees can
be characterized by a $\HCTL$-formula, the spoiler has a winning strategy in the
game for $L$. We expect the converse to be true as well but do not
attempt to prove it as it is not needed for our purposes here.

Let $L$ be a set of (finite or infinite) trees. The \emph{\HCTL-game}
for $L$ is played by two players, the \emph{spoiler} and the
\emph{duplicator}. First, the spoiler picks a number $k$ which will be
the number of rounds in the core game. Afterwards, the duplicator chooses
two trees, $\calT\in L$ and $\calT'\not\in L$. The goal of the spoiler
is to make use of the difference between $\calT$ and $\calT'$ in the
core game. 

The \emph{core game} consists of $k$ rounds of moves, where
in each round $i$ a node from $\calT$ and a node from $\calT'$
are selected according to the following rules.
The spoiler can choose whether she starts her move in $\calT$ or in
$\calT'$ and whether she plays a node move or a path
  move.

In a \emph{node move} she simply picks a node from $\calT$ (or $\calT'$) and
the duplicator picks a node in the other tree. We refer to these two
nodes by $a_i$ (in $\calT$)  and $a'_i$ (in $\calT'$), respectively, where $i$ is the number of the round. 

In a \emph{path move}, the spoiler first chooses one of the trees. Let
us assume she chooses $\calT$, the case of $\calT'$ is completely
analogous. She picks an already selected node $a_j$ of $\calT$, for
some $j<i$ and a path $\pi$ starting in $a_j$. However, a node $a_j$
can only be selected if there is no other node $a_l$, $l<i$ below $a_j$.
The duplicator answers by
selecting a path $\pi'$ from $a'_j$. Then, the spoiler selects some node
$a'_i$ from $\pi'$ and the duplicator selects a node $a_i$ from $\pi$.

The duplicator wins the game if at the end the following conditions
hold, for every $i,j\le k$:
\begin{itemize}
\item $a_i$ is the root iff $a'_i$ is the root;
\item $a_i=a_j$ iff $a'_i=a'_j$;
\item for every proposition $p$, $p$ holds in $a_i$ iff it holds in $a'_i$;
\item there is a (downward) path from $a_i$ to $a_j$ iff there is a path from
  $a'_i$ to $a'_j$;
\item $a_j$ is a child of $a_i$ iff $a'_j$ is a child of $a'_i$.
\end{itemize}

\begin{theorem}\label{theo:ef}
  If a set $L$ of (finite and infinite) trees can be characterized by a \HCTL-formula, the
  spoiler has a winning strategy on the \HCTL-game for $L$.
\end{theorem}
\mfcsorfull{
The proof of Thm.~\ref{theo:ef} is by induction on the structure of
the \HCTL-formula \cite{KaraLSW09}.
}
{
\begin{proof}
   Let $L$ be a set of trees and $\varphi\in\HCTL$ such that, for every
  tree $\calT$, $\calT$ is in $L$ if and only if $\calT\models
  \varphi$. We show that the spoiler has a winning
  strategy with $k_\varphi$ rounds in the game for $L$, where
  $k_\varphi$ only depends on $\varphi$.

The proof is by induction on the structure of $\varphi$. 
As usual, we have to prove a slightly stronger statement for the
induction step. We show that, for every $\HCTL$-formula $\varphi$ with
variables from $X_l:=\{x_1,\ldots,x_l\}$, there
is $k_\varphi$ such that, for trees $\calT, \calT'$, nodes $v$ from
$\calT$ and $v'$ from $\calT'$ and node vectors  $\nvec$
and $\nvec'$, the spoiler has a winning strategy in the
$k_\varphi$-round core game on $(\calT,v,\nvec)$ and $(\calT',v',\nvec')$ if
$\calT,v,\nvec\models\varphi$ and $\calT',v',\nvec'\not\models\varphi$. 

Thus, the proof uses a slightly extended game, in which the duplicator does not only
choose $\calT$ and $\calT'$ but
also nodes $v,v'$ and node vectors $\nvec,\nvec'$. The game starts in
a situation where
nodes $a_0:=v$ and  $a_i:=\com_i$, for $1\le i\le l$ are
already selected in $\calT$ and correspondingly in $\calT'$. 
In the remaining $k$ rounds $a_{l+1},\ldots,a_{l+k}$ and
$a'_{l+1},\ldots,a'_{l+k}$ are selected and the winning condition
applies to $a_0,\ldots,a_{l+k}$ and $a'_0,\ldots,a'_{l+k}$.

It is easy to see that the theorem follows from this extended statement.

If $\varphi$
is atomic, it can only test propositional properties of $v$ and
hence the spoiler wins the game by choosing $k_\varphi=0$. 

The rest of the proof is by case distinction on the outermost
operator or quantifier of $\varphi$.
\begin{itemize}
\item If $\varphi=\neg \psi$, the spoiler has a winning strategy for $\psi$ by
  the hypothesis. She can simply follow that winning strategy whilst 
  switching the roles of $\calT$ and $\calT'$. In particular, $k_\varphi=k_\psi$.
\item If $\varphi=\psi\lor\chi$ the spoiler chooses
  $k_\varphi=\max(k_\psi,k_\chi)$. In the core game she either follows
  the winning strategy for $\psi$ or for $\chi$ depending on whether
  $\calT,v,\nvec\models \psi$ or  $\calT,v,\nvec\models \chi$.
\item The case that $\varphi=\psi\land\chi$ is analogous to the
  previous one.
\item If $\varphi=\E\X\psi$ the spoiler chooses
  $k_\varphi=k_\psi+1$. Let $\calT,v,\nvec,\calT',v',\nvec'$ be selected by
  the duplicator. As $\calT,v,\nvec\models\varphi$ there is a child $w$ of
  $v$ such that $\calT,w,\nvec\models\psi$. On the other hand, there is no
  child $w'$ of $v'$ with $\calT',w',\nvec'\models\psi$. Thus, the spoiler
  can select $a_{l+1}:=w$ and win the remaining $k_\psi$ rounds no matter which child of $v'$
  is chosen by the duplicator. She simply mimics the strategy of the
  game for $\psi$ on $(\calT,a_{l+1},\nvec)$ and $(\calT',a'_{l+1},\nvec')$.
If the duplicator does not choose a
  child of $v'$ the spoiler wins instantly.
\item As $\A\X\psi\equiv\neg\E\X\neg\psi$, the case of
  $\varphi=\A\X\psi$ is already covered by the previous cases.
\item  If $\varphi=\E(\psi\U\chi)$ the spoiler chooses
  $k_\varphi=\max(k_\psi+2,k_\chi+1)$. Let $\calT,v,\nvec,\calT',v',\nvec'$ be selected by
  the duplicator. As $\calT,v,\nvec\models \E(\psi\U\chi)$, there is some
  node $w$ below $v$ such that  $\calT,w,\nvec\models\chi$ and, for each
  node $z$ on the path from $v$ to $w$ it holds
  $\calT,z,\nvec\models\psi$. 
The spoiler does a node move in $\calT$ and
  selects $a_{l+1}=w$. 

Let $w'$ be the node selected by the
  duplicator. As $\calT',v',\nvec'\not\models\varphi$, we can conclude
  that $\calT',w',\nvec'\not\models\chi$ or, for some $z'$ on the path
  from $v'$ to $w'$, $\calT',z',\nvec'\not\models\psi$.  In the former
  case, the game continues, in the latter case she
  selects $v'_{l+2}=z'$. In either case, she has a winning strategy
  for the remaining $\max(k_\psi,k_\chi)$ rounds by induction.
\item If $\varphi=\A(\psi\U\chi)$, $\varphi$ is equivalent to $\varphi_1\land\varphi_2$ where
  $\varphi_1=\neg\E\G\neg\chi$ and $\varphi_2=\neg\E((\neg
  \chi)\U(\neg\psi\land\neg\chi))$.  The
  spoiler chooses $k_\varphi= \max(k_\psi,k_\chi)+2$.

Let $\calT,v,\nvec,\calT',v',\nvec'$ be selected by
  the duplicator. If $\calT',v',\nvec'\not\models\varphi_2$ the winning
  strategy of the spoiler is already given by the previous
  cases. Otherwise, $\calT',v',\nvec'\models\E\G\neg\chi$. Let $\rho'$ be
  a path starting from $v'$ such that $\calT',\rho',0\models
  \G\neg\chi$. Let $w'$ be the first node on this path for which none
of the nodes  $a'_1,\ldots,a_l$ is below $w'$. The spoiler selects
$w'$ in a node move. Let $w$ be the node selected by the
duplicator. If there is a node $z$ on the path from $v$ to $w$ such
that $\calT,z,\nvec\models\chi$, the spoiler chooses $z$ in a subsequent
node move and wins by induction as there is no corresponding node
between $v'$ and $w'$. Otherwise she makes a path move in which she
first selects the sub-path $\pi'$ of $\rho'$ starting in $w'$. 
 Let  $\pi$ be a path in $\calT$ starting from $w$ as selected by the duplicator. As
  $\calT,w,\nvec\models\varphi_1$, there is a node $z$ on $\pi$ such that
  $\calT,z,\nvec\models\chi$. The spoiler selects this node as $a_{i+1}$
    and, as the duplicator cannot find such a node on $\pi'$ wins by induction. 
  \item If $\varphi=\dx_i\psi$, for some $i$, the spoiler simply chooses
    $k_\varphi=k_\psi$. After the selection of
    $\calT,v,\nvec,\calT',v',\nvec'$ by the duplicator she mimics the game for
    $\psi$ on  the structures  $(\calT,v,\nvec[i/v])$ and $(\calT',v',\nvec'[i/v'])$. 
\item If $\varphi=\ax[i]\psi$, for some $i$, the spoiler mimics the game on
 $(\calT,\com_i,\nvec)$ and $(\calT',\com'_i,\nvec')$.
\end{itemize}
\qed
\end{proof}
} 

\ignore{
For a structure $\calA$ and nodes $a,a'$ in $\calA$, the
\emph{$k$-type} $\tau_k(\calA,a_x,a_y)$ is the set of formulas $\varphi$
of nesting depth at most $k$, for which
$\calA,a_x,a_y\models\varphi$. The following proposition is standard but
will play an important role in our proofs.
\begin{proposition}
  For every $k$, there are only finitely many $k$-types.
\end{proposition}
The proof is by induction on $k$. It uses that (1) from
finitely many formulas only finitely many formulas can be applied by
applying one formation rule of \Hone and (2) that from a finite set of
formulas only finitely many non-equivalent formulas can be obtained by
Boolean combination.
}

Now we turn to the proof of Thm.~\ref{theo:fairness}. It makes use of the 
following lemma which is easy to prove using standard techniques (see,
e.g., \cite{Libkin04}). The lemma will be used to show that the
duplicator has certain move options on paths starting from the
root. The parameter $S_k$ given by the lemma will be used below for the
construction of the structures $\calB_k$. 

For a string $s \in \Sigma^*$ and a symbol $a \in \Sigma$ let $|s|$
denote the length of $s$ and $|s|_a$ the number of occurrences of $a$ in $s$.

\begin{lemma}\label{lem:zeroone}
  For each $k\ge 0$ there is a number $S_k\ge 0$ such that, for each
  $s \in \{0,1\}^*$ there is an $s' \in \{0,1\}^*$ such that $|s'| \le S_k$ and $s\equiv_k s'$.
\end{lemma}
Here, $\equiv_k$ is equivalence with respect to the $k$-round
Ehrenfeucht game on strings (or equivalently with respect to
first-order sentences of quantifier depth $k$). It should be noted
that, if $k\ge 3$ and $s\equiv_k s'$, then the following conditions hold.
\begin{itemize}
\item $s\in\{0\}^*$ implies $s'\in\{0\}^*$.
\item If the first symbol of $s$ is $1$ the same holds for $s'$.
\item If $s$ does not have consecutive 1's, $s'$ does not either.
\end{itemize}

We fix some $S_k$, for
each $k$. 

The proof of  Thm.~\ref{theo:fairness} uses the \HCTL-game defined above. Remember
that the spoiler opens the game with the choice of a $k \in \mathbb{N}$ and the
duplicator responds with two trees $\mathcal{T} \in L$ and $\mathcal{T'} \not\in L$.
We want to show that the duplicator has a winning strategy so we need to construct
such trees, and then need to show that the duplicator has a winning strategy for the
$k$-round core game on $\mathcal{T}$ and $\mathcal{T}'$.

We will use transition systems in order to finitely represent infinite trees.
A transition system is a $\mathcal{K} = (V,E,v_0,\ell)$ where $(V,E)$ is a directed
graph, $v_0 \in V$, and $\ell$ labels each state $v \in V$ with a finite set of propositions.
The \emph{unraveling} $T(\mathcal{K})$ is a tree with node set $V^+$ and root
$v_0$. A node $v_0\ldots v_{n-1} v_n$ is a child of $v_0\ldots v_{n-1}$
iff $(v_{n-1},v_n) \in E$. Finally, the label of a node $v_0\ldots v_n$ is
$\ell(v_n)$. 

Inspired by \cite{EmersonH86} we define transition systems $\calA_i$, for each $i\ge 0$, as
depicted in Fig.~\ref{fig:aibk} (a). Nodes in which
$p$ holds are depicted black, the others are white (and we
subsequently refer to them as black and white nodes, respectively).
\begin{figure}[h]
  \centering
\psset{unit=5mm}
  \begin{minipage}[c]{4cm}
\centering
\begin{pspicture}(0,-0.3)(7,3.2)
 \psset{arrows=->}
  \rput(0.2,2){\Large$\mathbf{\calA_0:}$}
  \bnode{1,2.5}{b0}
  \wnode{1,1.5}{w0}
  \ncline{b0}{w0}
  \nccircle[angleA=-90]{w0}{0.2}
  \rput(3.7,2){\Large$\mathbf{\calA_k:}$}
  \bnode{4.5,2.5}{bk}
  \wnode{4.5,1.5}{wk}
  \ncline{bk}{wk}
  \nccircle[angleA=-90]{wk}{0.2}
  \rput(5.5,2){ $\Rightarrow$}
  \rput(6.7,2){$\calA_{k-1}$}
\end{pspicture}
   (a) 
  \end{minipage}
\hspace{1cm}
\begin{minipage}[c]{4cm}
 \centering 
\begin{pspicture}(0,-0.3)(4.5,3.2)
 \psset{arrows=->}
  \rput(0.2,2){\Large$\mathbf{\calB_k:}$}
  \bnode{2,4}{bb}
  \wnode{2,3.2}{bw1}
  \wnode{2,2.4}{bw2}
  \wnode{2,0.8}{bw3}
  \wnode{2,0}{bw4}
  \pnode(2,1.9){p1}
  \pnode(2,1.3){p2}
  \ncline{bb}{bw1}
  \ncline{bw1}{bw2}
  \ncline{bw2}{p1}
  \ncline[linestyle=dotted,arrows=-]{p1}{p2}
  \ncline{p2}{bw3}
  \ncline{bw3}{bw4}
  \nccurve[angleA=120,angleB=-120]{bw4}{bb}
  \nccircle[angleA=-90]{bw4}{0.2}
 \rput(3,2){\large $\Rightarrow$}
  \rput(4,2){$\calA_{N_k}$}
\end{pspicture}\\
 (b)
\end{minipage}
  \caption{Illustration of the definition of (a) $\calA_k$ and (b)
    $\calB_k$. The path of white nodes in $\calB_k$ consists of $S_k$
    nodes. The double arrow $\Rightarrow$ indicates that every white
    node on the left is connected to every black node on the right.}
  \label{fig:aibk}
\end{figure}
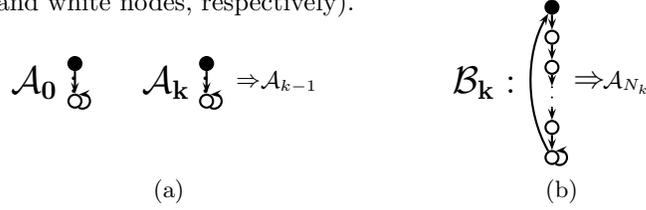

Thus, $\calA_0$ has a black (root) node and a white node with a
cycle. $\calA_i$ has a black (root) node, a white node with a cycle
and a copy of $\calA_{i-1}$. Furthermore, there is an edge from the white node below
the root of $\calA_i$ to each black node in the copy of $\calA_{i-1}$ (as
indicated by $\Rightarrow$). Let $\calT_i:=T(\calA_i)$. We first
introduce some notation and state some simple observations concerning
the tree $\calT_i$. 
\begin{enumerate}[(1)]
\item For a node $v$ in $\calT_i$ we denote the maximum number of
  black nodes on a path starting in $v$ (and not counting $v$ itself)
  the \emph{height} $h(v)$ of $v$. Then the root of $\calT_i$ has height
  $i$.
\item If $u$ and $v$ are black nodes of some $\calT_i$ with
  $h(u)=h(v)$ then the subtrees $T(u)$ and $T(v)$ induced by $u$ and $v$
  are isomorphic. 
\item The height of a tree is defined as the height of its root.
\item A white node $v$ of height $i$ has one white child (of height
  $i$) and $i$ black children of heights $0,\ldots,i-1$.
  A black node has exactly one white son.
\item Each finite path $\pi$ of $\calT_i$ induces a string
  $s(\pi) \in \{0,1\}^*$ in a natural way: $s(\pi)$ has one position, for each node
  of $\pi$, carrying a 1 iff the corresponding node is black. 
\item The root of $\calT_i$ has only one child. We call the subtree
  induced by this (white!) child $\calU_i$. If $v$ is a white node of height
  $i$ then $T(v)$ is isomorphic to $\calU_i$. 
\end{enumerate}
Next we define numbers $N_k$ inductively as follows: $N_0:=0$ and $N_k:=N_{k-1}+\max(S_3,S_k)+1$.

The following lemma shows that the duplicator has a winning strategy
in two structures of the same kind, provided they both have sufficient depth. 
\begin{lemma}
\label{lem:dupwinsTU}
  Let $i,j,k$ be numbers such that $i,j\ge N_k$. Then the duplicator
  has a winning strategy in the $k$-round core game on
\mfcsorfull{(a) $\calT_i$ and $\calT_j$, and (b) $\calU_i$ and $\calU_j$.
}
{
  \begin{enumerate}[(a)]
  \item   $\calT_i$ and $\calT_j$, and
  \item   $\calU_i$ and $\calU_j$.
  \end{enumerate}
}
\end{lemma}
\mfcsorfull{ 
\begin{proof}[Sketch]
  In both cases, the proof is by induction on $k$, the case $k=0$ being
  trivial. We consider (a) first. Let $k>0$ and let us assume that the
  spoiler chooses $v\in\calT_i$ in her first node move. We distinguish
  two cases based on the height of $v$.
  \begin{description}
  \item [$h(v)> N_{k-1}$:] Let $\pi$ denote the path from $r$ to $v$. By
Lemma~\ref{lem:zeroone} there is a string $s'$ with $|s'| \le S_l$ such that 
$s(\pi)\equiv_l s'$, where $l=\max(k,3)$. Here, $l\ge 3$
  guarantees in particular that $s'$ does not have consecutive 1's.
As $j\ge N_k =
N_{k-1}+S_l+1$,
there is a node
 $v'$ of height $\ge N_{k-1}$ in $\calT_j$ such that the path $\pi'$ from $r'$ to 
$v'$ satifies $s(\pi')=s'$. 
The duplicator 
chooses $v'$ as her answer in this round. By a compositional argument,
involving the induction hypothesis, it can be shown that the
duplicator has a winning strategy for the remaining $k-1$ rounds.
  \item [$h(v)\le N_{k-1}$:] Let $\pi$ be the path from $r$ to $v$, and $u_1$ be the 
highest black node on $\pi$ with $h(u_1) \le N_{k-1}$. Then we must have
$h(u_1) = N_{k-1}$ because $\pi$ contains black nodes of height up to 
$i \ge N_k$. Hence, $u_1$ has a white parent $u_2$ s.t.\ $h(u_2)>N_{k-1}$. We
determine a node $u_2'$ in $\calT'$ in the same way we picked $v'$ for $v$ in
the first case. In particular, $h(u'_2)\ge N_{k-1}$ and for the paths
$\rho$ leading from $r$ to $u_2$ and $\rho'$ leading from $r'$ to $u_2'$ we have 
$s(\rho)\equiv_k s(\rho')$. 

Let $u'_1$ be the black
child of $u'_2$ of height $h(u_1)$. As $h(u_1)=h(u'_1)$ there is an
isomorphism $\sigma$ between $T(u_1)$ and $T(u_2)$ and we choose
$v':=\sigma(v)$. 
An illustration is given in Figure \ref{fig:titj}.

The winning strategy of the duplicator for the remaining $k-1$ rounds
follows $\sigma$ on  $T(u_1)$ and $T(u_2)$ and is analogous to the
first case in the rest of the trees.
\mfcsonly{
The case of path moves is very similar, see \cite{KaraLSW09}.
}
  \end{description}
\qed
\end{proof}
}
{ 
\begin{myproofof}{Lemma \ref{lem:dupwinsTU}}
The proof is by induction on $k$, the case $k=0$ being trivial. Thus, let $k>0$. We 
first show (a). Let us assume first that the spoiler makes a {\bf node move} in $\calT_i$ 
on $v$ (node moves in $\calT_j$ are symmetric). We distinguish two cases depending on the 
height of $v$. In both cases let $r,r'$ be the roots of $\calT_i$ and $\calT_j$ respectively.

{\bf Case} $h(v)> N_{k-1}$. Let $\pi$ denote the path from $r$ to $v$. By
Lemma~\ref{lem:zeroone} there is a string $s'$ with $|s'| \le S_l$ such that 
$s(\pi)\equiv_l s'$, where $l=\max(k,3)$. As $j\ge N_k =
N_{k-1}+S_l+1$,
there is a node
 $v'$ of height $\ge N_{k-1}$ in $\calT_j$ such that the path\footnote{It should be noted that $l\ge 3$
  guarantees in particular that $s'$ does not have consecutive 1's.} $\pi'$ from $r'$ to 
$v'$ satifies $s(\pi')=s'$. The duplicator 
chooses $v'$ as her answer in this round. We have to show that she has a winning strategy for the
remaining $k-1$ rounds. Her strategy is a composition of the following
three strategies for different parts of the trees.\footnote{By a
  standard argument the different strategies can be combined into a
  strategy on $\calT$ and $\calT'$ (see again \cite{Libkin04}). It is
  helpful here that path moves only involve paths that are below all
  previously selected nodes. Hence, a path move always only affects
  one of the sub-games.}
\begin{enumerate}[(i)]
\item If the spoiler does a (node or path) move in $T(v)$ or $T(v')$
  the duplicator can reply according to her winning strategy in the
  $(k-1)$-round game in these two trees which is guaranteed by induction as both
  have height $\ge N_{k-1}$.
\item If the spoiler chooses a node on $\pi$ or $\pi'$ then the duplicator
  answers following her strategy in the $k$-round Ehrenfeucht game on the strings $s(\pi)$
  and $s(\pi')$.
\item The remaining (and most complicated) sub-strategy concerns
  moves elsewhere in $\calT_i$ (the case of a move elsewhere in $\calT'_j$
  is again symmetric).  Let $w$ be a node chosen by the spoiler (in a
  node move or as the starting node of a path). Let $y$ be the last node of $\pi$ on the 
  path $\rho$ from $r$ to $w$. Node $y$ has two different successors -- namely one on
  $\pi$ and one on $\rho$, hence it must be white by fact (4) above. Let $z$ be its
  successor on $\rho$. Let $y'$ be the node corresponding to $y$ on $\pi'$ as
  induced by the winning strategy of the duplicator on $s(\pi)$ and
  $s(\pi')$.
  \begin{itemize}
  \item If $z$ has height $<N_{k-1}$ let $z'$ be the unique
    (black!)\footnote{$z$ must be black as all nodes on $\pi$ have
      height $\ge N_{k-1}$ and only black nodes may have a smaller
      height than their parent.}
 child of $y'$ of the same height as $z$. By fact (2) above, $T(z)$ and $T(z')$ 
    are isomorphic and the duplicator has a winning
    strategy in these two subtrees induced by an isomorphism. 
     \item If $z$ has  height $\ge N_{k-1}$ let $z'$ be some child of $y'$
    (of the same color as $z$, not on $\pi'$, and with height $\ge N_{k-1}$. 
    (Note that $z$ and $z'$ can be both black or both white). By induction the 
    duplicator has a winning strategy on $T(z)$ and $T(z')$.
     \end{itemize}
\end{enumerate}

\noindent
{\bf Case} $h(v)\le N_{k-1}$. Let $\pi$ be the path from $r$ to $v$, and $u_1$ be the 
highest black node on $\pi$ with $h(u_1) \le N_{k-1}$. Then we must have
$h(u_1) = N_{k-1}$ because $\pi$ contains black nodes of height up to 
$i \ge N_k$. Hence, $u_1$ has a white parent $u_2$ s.t.\ $h(u_2)>N_{k-1}$. We
determine a node $u_2'$ in $\calT'$ in the same way we picked $v'$ for $v$ in
the first case. In particular, $h(u'_2)\ge N_{k-1}$ and for the paths
$\rho$ leading from $r$ to $u_2$ and $\rho'$ leading from $r'$ to $u_2'$ we have 
$s(\rho)\equiv_k s(\rho')$. 

Let $u'_1$ be the black
child of $u'_2$ of height $h(u_1)$. As $h(u_1)=h(u'_1)$ there is an
isomorphism $\sigma$ between $T(u_1)$ and $T(u_2)$ and we choose
$v':=\sigma(v)$. 
An illustration is given in Figure \ref{fig:titj}.

The winning strategy of the duplicator for the remaining $k-1$ rounds
follows $\sigma$ on  $T(u_1)$ and $T(u_2)$ and is analogous to the
first case in the rest of the trees.\\

Next, we assume that the first move of the spoiler is a {\bf path
  move} in $\calT$ (path moves in $\calT'$ are again symmetric). 

Let $\pi$ be the path chosen by the spoiler. Let $v$ be the lowest
black node of $\pi$. Let $v'$ be the node chosen by the duplicator had
the spoiler selected $v$ in a node move and let $\pi'$ be the path
from $r'$ to $v'$ extended by the unique infinite white path below $v'$. We can distinguish the same
two cases as for node moves. 

If the node $a'$ selected by the spoiler on $\pi'$ is from $T(v')$ (in case
1) or from $T(u'_1)$ (in case 2), the duplicator can choose a
corresponding node $a$ by the isomorphism. Otherwise,  the duplicator
replies by the node $a$ of $\rho$ induced by the $k$-round (!) winning strategy of the
duplicator on $s(\rho)$ and $s(\rho')$.  

This completes the inductive step for (a).

For (b) the proof is completely analogous.
This completes the proof of the lemma.\\
\end{myproofof}
}
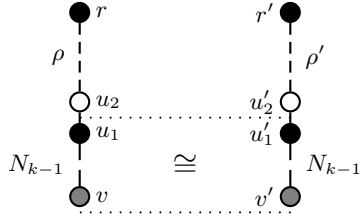
\begin{figure}[t]
  \centering
\psset{unit=7mm}
  \begin{pspicture}(6,3.5)
    \bnode{1,3.5}{r1}
\uput{0.3}[0](1,3.5){$r$}
    \wnode{1,1.8}{u21}
\uput{0.3}[0](1,1.8){$u_2$}
    \bnode{1,1.2}{u11}
\uput{0.3}[0](1,1.2){$u_1$}
    \gnode{1,0}{v1}
 \uput{0.3}[0](1,0){$v$}
   \ncline[linestyle=dashed]{r1}{u21}\nbput{$\rho$}
    \ncline{u21}{u11}
    \ncline[linestyle=dashed]{u11}{v1}\nbput{$N_{k-1}$}
    \bnode{5,3.5}{r2}
\uput{0.3}[180](5,3.5){$r'$}
    \wnode{5,1.8}{u22}
\uput{0.3}[180](5,1.8){$u_2'$}
    \bnode{5,1.2}{u12}
\uput{0.3}[180](5,1.2){$u_1'$}
    \gnode{5,0}{v2}
 \uput{0.3}[180](5,0){$v'$}
   \ncline[linestyle=dashed]{r2}{u22}\naput{$\rho'$}
    \ncline{u22}{u12}
    \ncline[linestyle=dashed]{u12}{v2}\naput{$N_{k-1}$}
   \psline[linestyle=dotted](1.,1.5)(5,1.5)
   \psline[linestyle=dotted](1.,-0.3)(5,-0.3)
   \rput(3,0.6){\large$\cong$}
  \end{pspicture}
  \caption{Illustration of the case where $h(v)\le N_{k-1}$. The
    colors of $v$ and $v'$ are not known a priori.}
  \label{fig:titj}
\end{figure}


We are now prepared to prove Thm.~\ref{theo:fairness}.
\mfcsorfull{
\begin{proof}[of Thm.~\ref{theo:fairness}]
By Thm.~\ref{sameExpPower} it is sufficient to show that no formula
equivalent to $\E\Finfty p$ exists in \HCTL. To this end, we prove that the duplicator has a winning strategy in the
\HCTL-game for the set of trees fulfilling $\E\Finfty p$.

We define transition systems $\calB_k$, for $k\ge 0$. As
illustrated in Figure \ref{fig:aibk} (b), $\calB_k$ has a black root
from which a path of length $S_k$  of white nodes starts. The last of
these white nodes has a self-loop and an edge back to the
root. Furthermore, $\calB_k$ has a copy of $\calA_{N_k}$ and there is
an edge from each white node of the initial path to each black node of
the copy of $\calA_{N_k}$. 
Clearly, for each $k$, $T(\calB_k)\models\E\Finfty p$ and 
$T(\calA_k)\not\models\E\Finfty p$. 

It can be shown that, for each $k$, the
duplicator has a  winning strategy in the $k$-round  core game on
$\calT=T(\calB_k)$ and $\calT'=T(\calA_{N_k})$ \cite{KaraLSW09}. 
\qed
\end{proof}
}
{ 
\begin{myproofof}{Theorem \ref{theo:fairness}}
By Theorem \ref{sameExpPower} it is sufficient to show that no formula
equivalent to $\E\Finfty p$ exists in \HCTL. To this end, we prove that the duplicator has a winning strategy in the
\HCTL-game for the set of trees fulfilling $\E\Finfty p$.

We next define transition systems $\calB_k$, for $k\ge 0$. As
illustrated in Figure \ref{fig:aibk} (b), $\calB_k$ has a black root
from which a path of length $S_k$  of white nodes starts. The last of
these white nodes has a self-loop and an edge back to the
root. Furthermore, $\calB_k$ has a copy of $\calA_{N_k}$ and there is
an edge from each white node of the initial path to each black node of
the copy of $\calA_{N_k}$. 

Clearly, for each $k$, $\calB_k\models\E\Finfty p$ and 
$\calA_k\not\models\E\Finfty p$. 

We show in the remainder of the proof that, for each $k$, the
duplicator has a  winning strategy in the $k$-round  core game on $\calT=T(\calB_k)$ and $\calT'=T(\calA_{N_k})$. 

We call the nodes of $\calT$ induced by the extra nodes that $\calB_k$
has over $\calA_{N_k}$ \emph{extra nodes}. 

The proof is again by induction on $k$ and it is very similar to the
proof of Lemma \ref{lem:dupwinsTU}. The case $k=0$ is again trivial.

A node move $v$ in $\calT$ is answered just as it was in the proof of
the lemma. The duplicator then wins by induction. Likewise node moves $v'$ in $\calT'$ are
answered as in the proof of the lemma and (besides the root) no special black
node of $\calT$ is involved. 

It remains to deal with path moves. Path moves starting in $\calT'$
can be handled as in the proof of the lemma. Likewise path moves starting in
$\calT$ with a path $\pi$ of finitely many black nodes can be handled
along the same lines.

The only real new case is when the spoiler chooses a path $\pi$ in
$\calT$ that has infinitely many black nodes. The duplicator answers
by choosing the unique path $\pi'$ of $\calT'$ starting from $r'$ that only
consists of white nodes. Let $v'$ be a node of $\pi'$ that is
chosen by the spoiler. By (the remark after) Lemma \ref{lem:zeroone},  there is a
string $s$ of $\le S_k$ zeros such that $s\equiv_k s(\rho')$, where $\rho'$ is
the path from $r'$ to $v'$. The duplicator thus picks the node $v$ on
$\pi$ such that $s(\rho)=s$, where $\rho$ denotes the path from $r$ to
$v$. Note that by construction, the initial white path of $\calT$ is
long enough to guarantee the existence of such a node $v$. That the
duplicator has a winning strategy for the remaining $k-1$ rounds can
be shown along the same lines as before. Subsequent choices of paths with
infinitely many black nodes are answered by paths with one black node
and infinitely many white ones.
\end{myproofof}
}



\section{Satisfiability of \Honeplus} \label{sec:complexity}
\begin{theorem}\label{compl-lower-bound}
  Satisfiability of {\Honeplus} is hard for {\Texptime}. 
\end{theorem}
\begin{proof}
The proof is by reduction from a tiling game (with \Texptime complexity) to the satisfiability problem of
\Honeplus. Actually we show that the lower bound even holds for the fragment of \Honeplus without the
$\U$-operator (but with the $\F$-operator instead).

An instance $I=(T,H,V,F,L,n)$ of the \emph{2EXP-corridor tiling game} consists of a finite set $T$ of \emph{tile
types}, two relations $H,V \subseteq T \times T$ which constitute the \emph{horizontal and vertical
constraints}, respectively, two sets $F,L \subseteq T$ which describe
the starting and end conditions, respectively, and a number $n$
given in unary. The game is played by two players, $E$ and $A$, on a board consisting of $2^{2^n}$ columns
and (potentially) infinitely many rows. Starting with player $E$ and following the constraints $H$, $V$ and $F$ the
players put tiles to the board consecutively from left to right and row by row. The constraints prescribe the following conditions:
\begin{itemize}
\item A tile $t'$ can only be placed immediately to the right of a tile $t$ if $(t,t') \in H$.
\item A tile $t'$ can only be placed immediately above a tile $t$ if $(t,t') \in V$.
\item The types of all tiles in the first row belong to the set $F$. 
\end{itemize}
Player $E$ wins the game if a row is completed containing only tiles
from $L$ or if $A$ makes a move that violates the constraints. On the
other hand, player $A$ wins if $E$ makes a forbidden move or the game goes on ad infinitum.

A winning strategy for $E$ has to yield a countermove for all possible moves
of $A$ in all possible reachable situations. Furthermore, the starting condition and the horizontal and vertical constraints have to be respected.
Finally, the winning strategy must guarantee that either 
player $A$ comes into a situation where he can no longer make an
allowed move or a row with tiles from $L$ is completed.

The problem to decide for an instance $I$ whether player $E$ has a winning strategy on $I$ is complete for 
\Texptime. This follows by a straightforward extension of \cite{Chlebus86}.

\fullonly{
Now we show in full detail how to build a formula $\varphi_{I}$ of
length $\bigO(|I|\cdot|T|)$ from an instance $I$ with tile set $T$ of the 2EXP-corridor tiling
game such that $\varphi_{I}$ is satisfiable if and only if player $E$
has a winning strategy in the game for $I$. In fact, a tree will
satisfy $\varphi_I$ if and only if it encodes a winning strategy of
player $E$. Here, the encoding tree represents all possible plays (for
the various moves of player $A$) for a fixed (and winning) strategy of
player $E$.

\paragraph*{Encoding of the winning strategy for player $E$.} 
We encode  strategies for player $E$ as  $T$-labeled trees in which each move is
represented by a sequence of nodes (see Figure \ref{FullStrategy}). The first move of player $E$ is represented by a sequence
starting at the root. Each sequence corresponding to a move of $E$ is followed by several branches, one
for every possible next move of $A$. Each sequence
	corresponding to a move of $A$ is followed by one sequence of nodes
corresponding to the move of player $E$ following the strategy. It is clear that
such a tree represents a \emph{winning} strategy if every root path
corresponds to a sequence of moves resulting in a win for player
$E$.

In the encoding we use the propositions $\{\pose, \poso,b_0,...,b_{n-1},\rowe, \rowo,b,$
$o,c,q_{\sharp}\} \cup \{p_t \ |\ t \in T\}$.  In order to be able to describe the constraints via a
\Honeplus-formula of polynomial length we serially number all positions of a row of the board in the
style of \cite{VardiS85}. While in \cite{VardiS85} a
row\footnote{Actually, the proof in \cite{VardiS85} uses alternating
  Turing machines and thus encodes configurations rather than rows.} consists of $2^n$ positions we have to
deal with $2^{2^n}$ positions in the current proof. We encode each
position by a sequence of $2^n$ nodes, each of which we call \emph{position
bits}. For each of these nodes the
propositions $b_0,...,b_{n-1}$ encode a binary number. Each position
bit in turn represents one bit of a binary
number of length $2^n$ via proposition $b$\footnote{It should be noted that the \emph{lowest} bit of this binary number is represented by the position bit with the \emph{highest}
number.}. Each such sequence is
preceded by a \emph{position node} which holds some additional information that will be described later. We call a sequence of length $2^n+1$ representing one position of
the tiling a \emph{position sequence}. In each position bit of a position
sequence proposition $p_t$ holds for the tiling type $t$ of its tile.  A row is represented by
$2^{2^n}$ position sequences preceded by a \emph{row node}.  
Altogether, a row is represented by a \emph{row sequence} consisting of $(2^n+1)2^{2^n}+1$
nodes.

For technical reasons, each position node of
an even (odd) position is marked using the proposition
$\pose$, ($\poso$, respectively). Likewise, the row nodes of
even (odd) rows are marked using $\rowe$ ($\rowo$). It is worth noting
that the tree branches only\footnote{After the modification in the
  next paragraph, this statement only holds with respect
to \emph{original nodes}.} after position nodes in which $\poso$ holds
as the odd positions are tiled by $A$.

%


To compare two nodes of a path that are far apart  we use a technique that was originally invented in \cite{VardiS85} and was also applied 
in \cite{JohannsenL03}. To this end, we use two kinds of nodes: \emph{original nodes} which are
labelled with the proposition $o$ and \emph{copy nodes} labelled with the proposition $c$.  For the
encoding of the winning strategy only the original nodes are relevant. Each original node has a copy
node with identical propositions (except for the proposition $o$) as a child. Likewise, each copy node
has only copy nodes with identical propositions as children (see Figure~\ref{rowwinStrategy}). Copy nodes
will enable us to ''mark'' an original node $v$ by fixing a path $\pi$ through $v$ and its copy node
child. Since copy nodes carry the propositions of their parent
original nodes, assertions about an original node can be tested in any
of their subsequent copy nodes.

Proposition $q_{\sharp}$ is used to label the part of the tree which does not belong to the encoding of the winning strategy. 
\paragraph{Testing vertical constraints.} 
The most difficult condition to test is that a tree respects the
vertical constraints. Thus, we have to check the following condition:
\emph{
  For every row,
  except the first one, the tile of every position is consistent with the
  tile of the corresponding position in the previous row.}

To this end, we have to compare two position sequences representing corresponding positions in consecutive rows. Two 
sequences represent corresponding positions if,
whenever two position bits are equivalent with respect to
$b_0,\ldots,b_{n-1}$, they are also equivalent with respect to
$b$. The technical challenge is to do this comparison with a formula
of linear (as opposed to exponential) size.
Finally, it has to be checked that the position bits of the two
position sequences are consistent with respect to $V$.

Let $r$ and $r'$ be row sequences representing two consecutive rows.
Let $s'$ be some position
sequence of $r'$ for which consistency with the corresponding position
of the previous row $r$
shall be checked. 

We first give an informal description of the formula that checks the
vertical constraints. Assume that $x$ is associated with the last
position bit of the position sequence $s'$ representing the $j$-th
position in $r'$. We first construct a formula $\xi$ that becomes
true exactly at the first position bit $v$ of the position sequence $s$
representing position $j$ in  row $r$ on the path to $x$. 

To this end, $\xi$ checks that there is a path starting in $v$,
continuing at least to the last node of $s$ (from where it might follow copy nodes) and from each non-copy node $u$ on this path, there is
a path leading to  some copy node of a node $u'$ in $s'$ with exactly the
same propositions $b_0,\ldots,b_{n-1},b$. 

This can be expressed as
\begin{align*}
  \label{eq:1}
\xi=\ &o \wedge \bigwedge_{i=0}^{n-1} \neg b_i \land\E\big[ \F \bigwedge_{i=0}^{n-1} b_i \land \G\big(o \rightarrow 
\E[(\F \rowo \wedge \neg\F \rowe \vee \F \rowe \wedge \neg\F \rowo)\land \\
&\G(\neg c \rightarrow \E\F x) \wedge \F\E(\F x \land \G\neg\pos)\land \bigwedge_{i=0}^{n-1} (b_i \leftrightarrow \F(c \wedge b_i)) \wedge b \leftrightarrow \F(c \wedge b)] \big)\big]
\end{align*}
Here, $\pos$ is an abbreviation for $\poso\lor\pose$ indicating that a
node is a position node. The path formula $\F \bigwedge_{i=0}^{n-1} b_i$ 
ensures that the current path continues at least to the last position bit of the current position from where
it might follow copy nodes. The path formula $\F \rowo \wedge \neg\F \rowe \vee \F \rowe \wedge \neg\F \rowo$ 
makes sure that the current path meets exactly two rows. The
path formula               
$\G(\neg c \rightarrow \E\F x) \wedge \F\E(\F x \land \G\neg\pos)$
tests that the path reaches
the same position sequence as the node carrying $x$ but no subsequent
position sequence.

The vertical constraints now hold if, whenever $x$ is put to the last
position node of some position sequence $s'$ with tile type $t'$, if
at some node $v$ the formula $\xi$ holds then the tile type $t$ at $v$
has to be such that $(t,t')\in V$. This can now be expressed by
$\A\G\big[\bigwedge_{t'\in T}[(o\land\bigwedge_{i=0}^{n-1}b_i\land
  p_{t'})\rightarrow\dx.\ar\A\G(\xi\rightarrow\bigvee_{(t,t')\in V}p_t)]\big]$.

For an instance $I$ of 2EXP-corridor tiling game we present the whole
formula $\varphi_I$ of length in $\bigO(|I||T|)$ such that $\varphi_I$ is satisfiable if and only if player $E$ has a winning strategy in the game for $I$. The formula $\varphi_I$ is composed of the conjunction of $\chi = \bigwedge_{i=1}^{10} \chi_i$ and $\psi = \bigwedge_{i=1}^7 \psi_i$ where 
$\chi$ describes the basic tree structure that is needed to formulate a strategy and 
the $\psi$ guarantees that the model of $\varphi_I$ corresponds to
a winning strategy for player $E$. Each of the subformulas $\chi_i$,
$\psi_i$ is of length $\bigO(|I||T|)$.

We first introduce some abbreviations: 
\begin{itemize}
\item $\pos = \pose \vee \poso$ (\emph{pos}ition node)
\item $\row = \rowe \vee \rowo$ (\emph{row} node)
\item $\phifirst = o \wedge \bigwedge_{i=0}^{n-1} \neg b_i$
(\emph{first} (and original) node in a position sequence)
\item $\philast = \bigwedge_{i=0}^{n-1} b_i$
(\emph{last} node in a position sequence, not necessarily original)
\item $\psicur = \G \neg \pos \wedge \F(c \wedge
\philast)$ (the path extends exactly until the
  end of the position sequence and continues with copy nodes; in this
  sense it is a \emph{full} path)
\item $\psipos = (\F \pose \wedge \neg \F \poso) \vee (\F \poso \wedge
\neg \F \pose)$ (the path meets \emph{two} (consecutive) \emph{pos}ition
  sequences)
\item $\psirow = (\F \rowe \wedge \neg \F \rowo) \vee (\F \rowo
  \wedge \neg \F \rowe)$ (the path meets \emph{two} (consecutive) \emph{row}
  sequences)
\end{itemize}
The first formula helps to describe some properties in a simple way.
\newline
\ourcondition{T1}{Each node of the tree is labelled with exactly one of the propositions $\rowe$, $\rowo$, $\pose$, $\poso$, $q_{\sharp}$, $o$ and $c$.} 
\begin{align*}
\chi_1 = \A\G [&(\row \vee  \pos \vee q_{\sharp} \vee o \vee c) \wedge \\ 
&(\pose \rightarrow \neg \poso \wedge \neg \rowe \wedge \neg \rowo \wedge \neg q_{\sharp} \wedge \neg o \wedge \neg c) \wedge  \\
&(\poso \rightarrow \neg \pose \wedge \neg \rowe \wedge \neg \rowo \wedge \neg q_{\sharp} \wedge \neg o \wedge \neg c) \wedge  \\
&(\rowe \rightarrow \neg \pose \wedge \neg \poso \wedge \neg \rowo \wedge \neg q_{\sharp} \wedge \neg o \wedge \neg c) \wedge  \\
&(\rowo \rightarrow \neg \pose \wedge \neg \poso \wedge \neg \rowe \wedge \neg q_{\sharp} \wedge \neg o \wedge \neg c) \wedge  \\
&(q_{\sharp} \rightarrow \neg \pose \wedge \neg \poso \wedge \neg \rowe \wedge \neg \rowo \wedge \neg o \wedge \neg c) \wedge \\
&(o \rightarrow \neg \pose \wedge \neg \poso \wedge \neg \rowe \wedge \neg \rowo \wedge \neg q_{\sharp} \wedge \neg c) \wedge \\         
&(c \rightarrow \neg \pose \wedge \neg \poso \wedge \neg \rowe \wedge \neg \rowo \wedge  \neg q_{\sharp} \wedge \neg o)]  
\end{align*}
\ourcondition{T2}{The root induces with the proposition $\rowe$ the encoding of the first row and every node labelled with $\rowe$ or $\rowo$ has exactly one child labled with $\pose$ signalising the encoding of a new position.}
\begin{align*} 
\chi_2 = \rowe \wedge \A\G[\row \rightarrow \E\X (\pose \wedge \dx.@_{\myroot} \E\F(\E\X x \wedge \A\X x))]
\end{align*}
\ourcondition{T3}{Every child of a $\pose$-or $\poso$-node is labelled with the initial position bit number encoded by $b_0 \ldots b_{n-1}$ .}
\begin{align*}
\chi_3 = \A\G[\pos \rightarrow \A\X\phifirst]
\end{align*}
\ourcondition{T4}{As long as the last position bit of a position is
  not reached, the next node is labelled with the next position bit
  number.}\\
In order to keep the length of $\chi_4$ within the bound
$\bigO(|I||T|)$, we make use of additional propositions
$d_0,\ldots,d_{n-1}$ and $e_0,\ldots,e_{n-1}$. The idea is that $d_i=1$ iff
$b_j=1$, for all $j<i$ and that  $e_i=1$ iff
$b_j=0$, for all $j<i$. 
\begin{align*} 
\chi_{4} ={}  &\chi_{4a}\land\chi_{4b}\\
\chi_{4a} ={} &\A\G\big(d_0 \land  \bigwedge _{i=1}^{n-1} [d_i \leftrightarrow
(d_{i-1}\land b_{i-1})]  \land e_0 \land \bigwedge _{i=1}^{n-1} [e_i \leftrightarrow
(e_{i-1}\land \neg b_{i-1})]\big)\\ 
\chi_{4b} ={} & \A\G\big[(o \wedge \neg \philast)
\rightarrow \\
&\qquad \big( \dx.\E\X(o \wedge [(e_{n-1} \land b_{n-1} \land \ax \neg
b_{n-1})\lor(\neg e_{n-1} \land (b_{n-1}\leftrightarrow \ax b_{n-1}))]\land\\
&\qquad \qquad \qquad \bigwedge_{i=0}^{n-2} 
[(e_{i+1}\land\ax d_{i+1}) \lor (e_i \land b_i \land \ax \neg
b_i)\lor(\neg e_i \land (b_i\leftrightarrow \ax b_i))])\big)\big]\\
\end{align*}
\ourcondition{T5}{Each position bit has a child, which represents a copy of it. The nodes of a subtree rooted at a copy node are labelled exactly with the same propositions.}
\begin{align*} 
\chi_5 = \A\G[&(o \rightarrow \dx. \E\X (c \wedge \bigwedge_{i=0}^{n-1}(b_i \leftrightarrow @_x b_i) \wedge b \leftrightarrow @_x b)) \wedge \\
&(c \rightarrow \dx.\A\G (c \wedge \bigwedge_{i=0}^{n-1}(b_i \leftrightarrow @_x b_i) \wedge b \leftrightarrow @_x b))]
\end{align*}
\ourcondition{T6}{Each position bit has exactly two children. One of them is its copy node and the other one is:  
\begin{enumerate}[(a)]
\item the next position bit if the current node is not already the last position bit or
\item a node containing one of the propositions $\pose$, $\poso$, $\rowe$, $\rowo$ and $q_{\sharp}$, otherwise.
\end{enumerate}}
\begin{align*}
  \chi_6 = \chi_{6a}\land \chi_{6b}
\end{align*}

\begin{align*} 
\chi_{6a} = \A\G[(o \wedge \neg \philast) \rightarrow  &\E\X(o \wedge \dx. @_{root} \E\F(\E\X x \wedge \A\X(o \rightarrow x))) \wedge \\
                                                                 	    &\E\X(c \wedge \dx. @_{\myroot} \E\F(\E\X x \wedge \A\X(c \rightarrow x))) \wedge  \\
																																	    &\A\X(o \vee c)]
\end{align*}
\begin{align*}
\chi_{6b} = \A\G[(o \wedge \philast) \rightarrow &\E\X(c \wedge \dx. @_{\myroot}\E\F(\E\X x \wedge \A\X(c     			
                                                                \rightarrow x))) \wedge \\
																															 &\E\X(\neg c \wedge \dx. @_{\myroot}\E\F(\E\X x \wedge \A\X(\neg c     																																\rightarrow x))) \wedge \\
																															 &\A\X(c \vee \pos \vee \row \vee q_{\sharp})
\end{align*}
\ourcondition{T7}{With the propositions $\rowe$ and $\rowo$ the counting of the positions of the current row starts. This means that the next position gets the initial position number. Therefore the proposition $b$ is set to false in every position bit of this position.}
\begin{align*}
\chi_7 = \A\G[\row \rightarrow \E\X\A\X\E[\psicur \wedge \G \neg b]]
\end{align*}
\newline
Compared to the increasing of a position bit number the increasing of a position number is a little bit complicated because in the latter case the bits are distributed over several nodes. In addition, we have to account for the case that player $A$ cannot make a move without violating the horizontal or vertical constraints. In this case the game is over and $q_{\sharp}$-nodes follow only. We describe the increasing of a position number in three parts.
\newline
\begin{align*}
  \chi_8 = \chi_{8a}\land \chi_{8b}\land \chi_{8c}
\end{align*}
\ourcondition{T8a}{If the last position is reached then a new row is started or the game is over.}
\begin{align*}
\chi_{8a} &= \A\G[\phifirst \wedge \E[\psicur \wedge \G b] \rightarrow \E[\G \neg \pos \wedge \F(o \wedge \philast \wedge \E\X(\row \vee q_{\sharp}))]]		
\end{align*}
\ourcondition{T8b}{If the last position is not reached and it is the turn of player $E$ then the next position sequence definitely has to be encoded and it gets the next position number.}
\begin{align*}
\chi_{8b} =& \A\G[\phifirst \wedge \E[\G \neg \pos \wedge \F(c \wedge \philast \wedge b) ] \wedge \E[\psicur \wedge \F \neg b ] \rightarrow \theta]
\end{align*}
\textit{Find the highest bit which has to be flipped.} 
\begin{align*}
\theta =& \E[\psicur \wedge \F(o \wedge \neg b \wedge \E\X(\neg c \wedge (o \rightarrow \E[\psicur \wedge \G b])) \wedge \theta')]
\end{align*}
\textit{Flip the same bit in the next position sequence.}
\begin{align*}
\theta' =& \dx.[\theta'' \wedge \E[\psicur \wedge \F(o \wedge \philast \wedge \\ 
&\ \ \ \ \ \ \ \ \ \ \ \ \ \ \ \ \ \ \ \ \ \ \ \ \ \ \ \E\X( \pos \wedge \A\X\E[\psicur \wedge \F(o \wedge \bigwedge_{i=0}^{n-1} (b_i \leftrightarrow @_x b_i) \wedge (b \leftrightarrow @_x b) \wedge \\ 
&\ \ \ \ \ \ \ \ \ \ \ \ \ \ \ \ \ \ \ \ \ \ \ \ \ \ \ \ \ \ \ \ \ \ \ \ \ \ \ \ \ \ \ \ \ \ \ \ \ \ \ \ \ \ \ \ \ \ \ \ \ \E\X(\neg c \wedge (o \rightarrow \E[\psicur \wedge \G \neg b])))]))]] 
\end{align*}
\textit{If the proposition $b$ is satisfied in a position bit preceding the flipped bits then it is also satisfied in the same position bit in the next position sequence.}
\begin{align*}
\theta'' = & @_{\myroot}\E\F[\phifirst \wedge \E[\G\neg \pos \wedge \F x \wedge \F(c \wedge \bigwedge_{i=0}^{n-1} (b_i \leftrightarrow @_x b_i)) \wedge \\
&\ \ \ \ \ \ \ \ \ \ \ \ \ \ \ \ \ \ \ \ \ \ \ \ \ \ \ \G(o \wedge \neg x \rightarrow \A[\psipos \wedge \bigwedge_{i=0}^{n-1} (b_i \leftrightarrow \F(c \wedge b_i)) \rightarrow b \leftrightarrow \F(c \wedge b)])]]
\end{align*}
\ourcondition{T8c}{If the last position is not reached and it is the turn of player $A$ then the next position sequence has not to be encoded but the game could be over.}
\begin{align*}
\chi_{8c} = & \A\G[\phifirst \wedge \E[\G \neg \pos \wedge \F(c \wedge \philast \wedge \neg b)] \rightarrow (\E[\psicur \wedge \F(o \wedge \philast \wedge \E\X q_{\sharp})]\vee \theta)]
\end{align*}
\ourcondition{T9}{Determining whether a position/row has an even or uneven number.}
\begin{align*}
  \chi_9 = \chi_{9a}\land \chi_{9b}
\end{align*}
\begin{align*}
\chi_{9a} = & \A\G[\pos \rightarrow \dx.@_{\myroot}\A\G[\pos \wedge \E\X\E[\psicur \wedge \F(o \wedge \philast \wedge \E\X x)] \rightarrow \\
&\ \ \ \ \ \ \ \ \ \ \ \ \ \ \ \ \ \ \ \ \ \ \ \ \ \ \ \ \ \ \ \ \ \ \ \ \ \ \ \ \ \ \ \ \ \ \ (\pose \leftrightarrow @_x \poso)]] \\ \\
\chi_{9b} = & \A\G[\row \rightarrow \dx.@_{\myroot}\A\G[\row \wedge \E\X\E[\G\neg \row \wedge \F(c \wedge \philast) \wedge \\
&\ \ \ \ \ \ \ \ \ \ \ \ \ \ \ \ \ \ \ \ \ \ \ \ \ \ \ \ \ \ \ \ \ \ \ \ \ \ \ \ \ \ \ \ \ \ \ \ \F(o \wedge \philast \wedge \E\X x)] \rightarrow (\rowe \leftrightarrow @_x \rowo)]] \\
\end{align*}
\ourcondition{T10}{Each move of player $A$ is followed by exactly one counter move of player $E$. Because the even positions correspond to the moves of player $E$, every position node labeled with $\pose$ must have exactly one child.}
\begin{align*}
  \chi_{10} = \A\G[\pose \rightarrow \E\X\dx.\ar\E\F(\E\X x \wedge \A\X x)]
\end{align*}

Now we use the tree structure described above to encode a winning strategy for player $E$.
\newline
\ourcondition{W1}{The game ends after a finite sequence of moves.}
\newline
We express this by postulating that on every rootpath which does not contain any copy nodes a $q_{\sharp}$-node is eventually reached. A $q_{\sharp}$-node is followed only by $q_{\sharp}$-nodes.
\begin{align*}
\psi_1 = & \A(\G \neg c \rightarrow \F q_{\sharp}) \wedge \A\G(q_{\sharp} \rightarrow \A\G q_{\sharp}) 
\end{align*}
\ourcondition{W2}{To each position belongs exactly one tile type.}
\newline
We remember that a tile type $t$ is represented by the proposition $p_t$ in all position bits of a position sequence.
\begin{align*}
\psi_2 = & \A\G[[o \rightarrow \bigvee_{t \in T}(p_t \wedge \bigwedge_{t \neq t'\in T}\neg p_{t'} )] \wedge [o \wedge \neg \philast \rightarrow \bigwedge_{t \in T} (t \leftrightarrow EX(o \wedge t))]]
\end{align*}
\ourcondition{W3}{According to the horizontal constraints, the tile type of every position, except the first one on each row, is consistent with the tile type of the precedent position.}
\begin{align*}
\psi_3 = & \A\G[\phifirst  \rightarrow \\
&\ \ \ \ \ \bigwedge_{t' \in T}[p_{t'} \rightarrow \dx.@_{\myroot}\A\G[\phifirst \wedge \neg x\wedge \E[\F x \wedge \psipos \wedge \G\neg \row]\rightarrow \bigvee_{(t,t')\in H} p_t]]]
\end{align*}
\ourcondition{W4}{According to the vertical constraints, for every row, except the first one, it holds that the tile type of every position on this row is consistent with the tile type of the same position on the precedent row.}
\begin{align*}
\psi_4 =&\A\G[\phifirst \rightarrow \bigwedge_{t' \in T}[p_{t'} \rightarrow \E[\psicur \wedge \F(o \wedge \philast \wedge \dx.@_{\myroot}\A\G[\xi \rightarrow \bigvee_{(t,t')\in V} p_t])]]] \\
\xi = &\phifirst \wedge \E[\F x \wedge \G(\neg c \rightarrow \E\F x) \wedge \psirow]\wedge \\
&\E[\psicur\wedge \G(o \rightarrow \E[\G(\neg c \rightarrow \E\F x)\wedge \F\E[\G\neg \pos \wedge \F x] \wedge \\
& \ \ \ \ \ \ \ \ \ \ \ \ \ \ \ \ \ \ \ \ \ \ \ \ \ \ \ \ \ \bigwedge_{i=0}^{n-1} (b_i \leftrightarrow \F(c \wedge b_i)) \wedge b \leftrightarrow \F(c \wedge b)]) ]
\end{align*}
\ourcondition{W5}{All possible moves of player $A$ are represented in the encoding.}
\begin{align*}
\psi_5 =& \A\G[\phifirst \wedge \E[\psicur \wedge \F(o \wedge \philast \wedge \neg b)] \rightarrow \\
&\ \ \ \ \ \ \bigwedge_{t \in T}(p_t \rightarrow \bigwedge_{(t,t')\in H}(\E[\psicur \wedge \F(o \wedge \philast \wedge \E\X(\pos \wedge \E\X p_{t'}))] \vee \\
&\ \ \ \ \ \ \ \ \ \ \ \ \ \ \ \ \ \ \ \ \ \ \ \ \ \ \ \ \ \E[\psicur \wedge \F(o \wedge \philast \wedge \\
&\ \ \ \ \ \ \ \ \ \ \ \ \ \ \ \ \ \ \ \ \ \ \ \ \ \ \ \ \ \ \ \ \ \ \ \ \ \ \ \ \ \ \ \ \ \dx.@_{\myroot}\E\F(\xi \wedge \\ 
&\ \ \ \ \ \ \ \ \ \ \ \ \ \ \ \ \ \ \ \ \ \ \ \ \ \ \ \ \ \ \ \ \ \ \ \ \ \ \ \ \ \ \ \ \ \ \ \ \ \ \ \ \ \ \ \ \ \ \bigvee_{(t'',t')\notin V}\E[\psicur \wedge \\
&\ \ \ \ \ \ \ \ \ \ \ \ \ \ \ \ \ \ \ \ \ \ \ \ \ \ \ \ \ \ \ \ \ \ \ \ \ \ \ \ \ \ \ \ \ \ \ \ \ \ \ \ \ \ \ \ \ \ \ \ \ \ \ \ \ \ \ \ \ \ \ \ \F(o \wedge \philast \wedge \E\X (p_{t''} \wedge \E\F x))]))]))] 
\end{align*}
\ourcondition{W6}{All tile types in the first row are from the set $F$.}
\begin{align*}
\psi_6 = \A\G[\phifirst \wedge \dx.@_{\myroot}\E\X\E[\G\neg \row \wedge \F x] \rightarrow \bigvee_{t \in F} p_t]
\end{align*}
\ourcondition{W7}{Unless the game terminates prematurely (because player $A$ is not able to make a further move) all tile types in the last row are from the set $L$.}
\begin{align*}
\psi_7 = & \A\G[o \wedge \philast \wedge b \wedge \E\X q_{\sharp} \rightarrow \\
&\ \ \ \ \ \dx.@_{\myroot}\E\F[\row \wedge \E\X\E[\G \neg \row \wedge \F x \wedge \G(\phifirst \rightarrow \bigvee_{t \in L} p_t)]]]
\end{align*}
\newline
\newline
Finally we obtain the formula $\varphi_I$ as a conjunction of the formulas encoding the properties T1-T10 and W1-W7. It should be noticed that a model for $\varphi_I$ can contain multiple possible moves for player $E$ on each position. In this case it suffices to choose one of the suggested moves because every rootpath without copy nodes represents a win for $E$.
}

\fullonly{
\begin{figure}[t]
\begin{center}
\scalebox{0.85}{
\begin{pspicture}(1.0,0.0)(15,2) 
\psframe[framearc=0.5, linestyle=dotted](1.1,-0.4)(4.6,1.3)
\rput(-0.6,0.8){\rnode{00}{$\ldots$}}
\wnode{0.0,0.8}{row1} 
\rput(0.0,0.4){\rnode{00}{$\rowe$}}
\wnode{0.7,0.8}{pos1} 
\rput(0.7,0.4){\rnode{00}{$\pose$}}
\wnode{1.4,0.8}{o1} 
\rput(1.4,0.4){\rnode{00}{\small{$0$}}}
\rput(1.4,0.0){\rnode{00}{\small{$0$}}}
\dashednode{2.3,1.8}{c1} 
\wnode{2.1,0.8}{o2} 
\rput(2.1,0.4){\rnode{00}{\small{$1$}}}
\rput(2.1,0.0){\rnode{00}{\small{$0$}}}
\dashednode{3.0,1.8}{c2} 
\pnode(2.5,0.8){o2d}
\pnode(3.1,0.8){o2e}
\rput(2.8,0.8){\rnode{00}{$\ldots$}}
\wnode{3.5,0.8}{o3} 
\rput(3.3,0.4){\rnode{00}{\small{$2^n{-}2$}}}
\rput(3.5,0.0){\rnode{00}{\small{$0$}}}
\dashednode{4.4,1.8}{c3} 
\wnode{4.2,0.8}{o4} 
\rput(4.2,0.4){\rnode{00}{\small{$2^n{-}1$}}}
\rput(4.2,0.0){\rnode{00}{\small{$0$}}}
\dashednode{5.1,1.8}{c4} 

\psframe[framearc=0.5, linestyle=dotted](5.3,-0.4)(8.8,1.3)
\wnode{4.9,0.8}{pos2} 
\rput(4.9,0.4){\rnode{00}{$\poso$}}
\wnode{5.6,0.8}{o5} 
\rput(5.6,0.4){\rnode{00}{\small{$0$}}}
\rput(5.6,0.0){\rnode{00}{\small{$0$}}}
\dashednode{6.5,1.8}{c5} 
\wnode{6.3,0.8}{o6} 
\rput(6.3,0.4){\rnode{00}{\small{$1$}}}
\rput(6.3,0.0){\rnode{00}{\small{$0$}}}
\dashednode{7.2,1.8}{c6} 
\pnode(6.7,0.8){o6d}
\pnode(7.3,0.8){o6e}
\rput(7.0,0.8){\rnode{00}{$\ldots$}}
\wnode{7.7,0.8}{o7} 
\rput(7.5,0.4){\rnode{00}{\small{$2^n{-}2$}}}
\rput(7.7,0.0){\rnode{00}{\small{$0$}}}
\dashednode{8.6,1.8}{c7} 
\wnode{8.4,0.8}{o8} 
\rput(8.4,0.4){\rnode{00}{\small{$2^n{-}1$}}}
\rput(8.4,0.0){\rnode{00}{\small{$1$}}}
\dashednode{9.3,1.8}{c8} 
\wnode{9.1,0.8}{pos3} 
\rput(9.1,0.4){\rnode{00}{$\pose$}}
\pnode(9.5,0.8){pos3d}
\pnode(10.8,0.8){pos3e}
\rput(9.8,0.8){\rnode{00}{$\ldots$}}
\rput(10.5,0.8){\rnode{00}{$\ldots$}}

\psframe[framearc=0.5, linestyle=dotted](11.6,-0.4)(14.5,1.3)

\wnode{11.2,0.8}{pos4} 
\rput(11.2,0.4){\rnode{00}{$\poso$}}
\wnode{11.9,0.8}{o9} 
\rput(11.9,0.4){\rnode{00}{\small{$0$}}}
\rput(11.9,0.0){\rnode{00}{\small{$1$}}}
\dashednode{12.8,1.8}{c9} 
\wnode{12.6,0.8}{o10} 
\rput(12.6,0.4){\rnode{00}{\small{$1$}}}
\rput(12.6,0.0){\rnode{00}{\small{$1$}}}
\dashednode{13.5,1.8}{c10} 
\pnode(13.0,0.8){o10d}
\pnode(13.6,0.8){o10e}
\rput(13.3,0.8){\rnode{00}{$\ldots$}}
\wnode{14.0,0.8}{o11} 
\rput(14.0,0.4){\rnode{00}{\small{$2^n{-}1$}}}
\rput(14.0,0.0){\rnode{00}{\small{$1$}}}
\dashednode{14.9,1.8}{c11} 

\wnode{14.9,0.8}{row2} 
\rput(14.9,0.4){\rnode{00}{$\rowo$}}
\wnode{15.6,0.8}{pos5} 
\rput(15.7,0.4){\rnode{00}{$\pose$}} 

\rput(16.4,0.8){\rnode{00}{$\ldots$}}


\ncline{->}{row1}{pos1}
\ncline{->}{pos1}{o1}
\ncline{->}{o1}{o2}
\ncline{->}{o2}{o2d}
\ncline{->}{o2e}{o3}
\ncline{->}{o3}{o4}
\ncline{->}{o4}{pos2}
\ncline{->}{pos2}{o5}
\ncline{->}{o5}{o6}
\ncline{->}{o6}{o6d}
\ncline{->}{o6e}{o7}
\ncline{->}{o7}{o8}
\ncline{->}{o8}{pos3}
\ncline{->}{pos3}{pos3d}
\ncline{->}{pos3e}{pos4}
\ncline{->}{pos4}{o9}
\ncline{->}{o9}{o10}
\ncline{->}{o10}{o10d}
\ncline{->}{o10e}{o11}
\ncline{->}{o11}{row2}
\ncline{->}{row2}{pos5}

\ncline{->}{o1}{c1}
\ncline{->}{o2}{c2}
\ncline{->}{o3}{c3}
\ncline{->}{o4}{c4}
\ncline{->}{o5}{c5}
\ncline{->}{o6}{c6}
\ncline{->}{o7}{c7}
\ncline{->}{o8}{c8}
\ncline{->}{o9}{c9}
\ncline{->}{o10}{c10}
\ncline{->}{o11}{c11}
\end{pspicture}
}
\end{center}

\label{rowwinStrategy}
\caption{A row sequence in the encoding of a winning strategy for
  player $E$. The upper nodes are copy nodes.}
\end{figure}
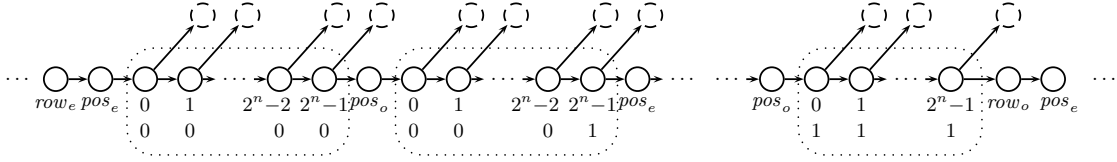
}

\fullonly{
\begin{figure}[t]
\begin{center}
\includegraphics[width=10cm]{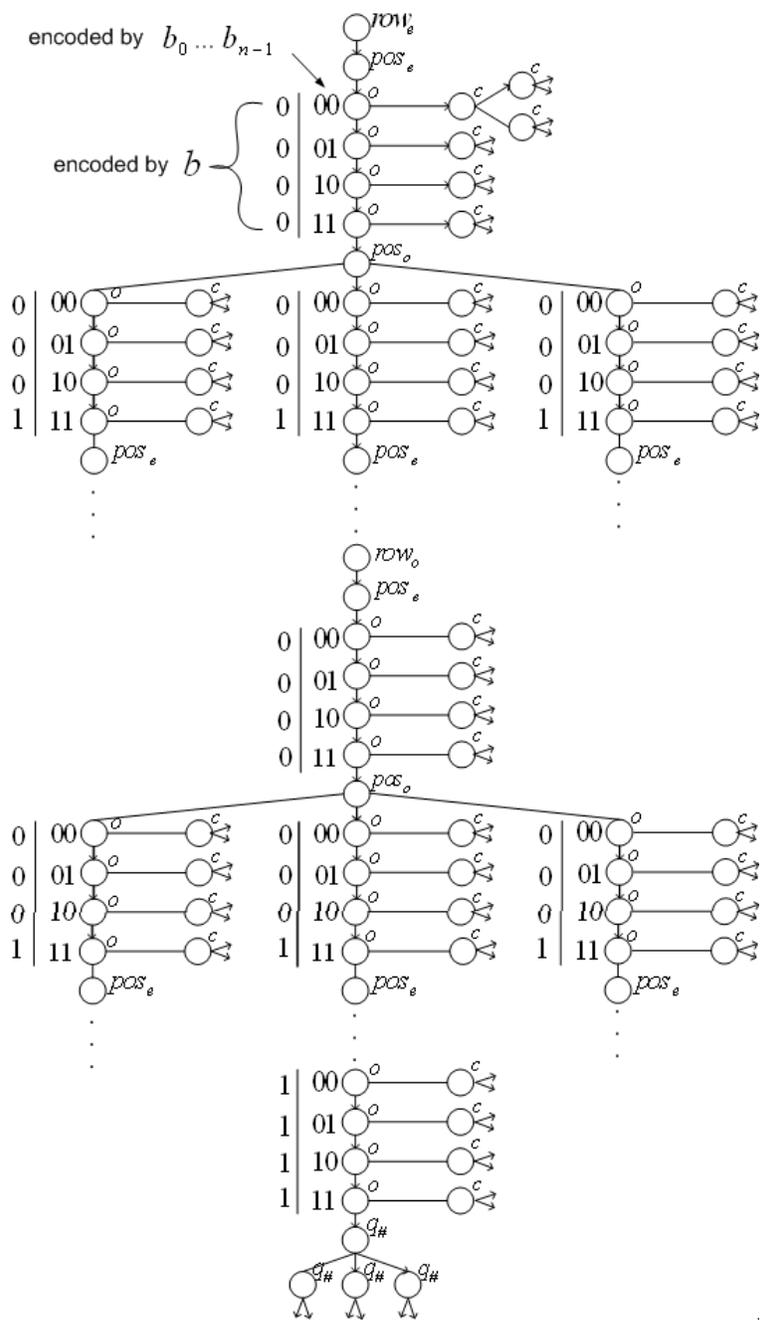}
\end{center}

\caption{An encoding of a winning strategy for
  player $E$ in the case where $n=2$.}
\label{FullStrategy}
\end{figure}
}
\qed
\end{proof}

We can obtain, by simple instantiation, a consequence of this lower complexity bound which will
be useful later on in proving the exponential succinctness of \Honeplus in \Hone.

\begin{corollary}
\label{cor:deepmodels}
There are finitely satisfiable \Honeplus formulas $\varphi_n$, $n \in \Nat$, of size $\bigO(n)$ s.t.\ 
every tree model $\mathcal{T}_n$ of $\varphi_n$ has height at least $2^{2^{2^n}}$.
\end{corollary}

\begin{proof} 
\mfcsorfull{ 
It is not difficult to construct instances $I_n$, $n \in \Nat$, of the 2EXP-tiling game 
with $|I| = \bigO(n)$ over a set $T$ of tiles with $|T| = \bigO(1)$ such that
player $E$ has a winning strategy and any successful tiling of the $2^{2^n}$-corridor requires 
$2^{2^{2^n}}$ rows. In order to achieve this, one encodes bits using tiles and forms the constraints in a
way that enforces the first row to encode the number $0$ in binary of length $2^{2^n}$, and each other row to
encode the successor in the natural number of the preceding row, while winning requires the number 
$2^{2^{2^n}}$ to be reached. The construction in the proof of Thm.~\ref{compl-lower-bound} then maps each 
such $I_n$ to a formula $\varphi_n$ of
size $\bigO(n)$ that is finitely satisfiable such that every finite model $\mathcal{T}_n$ of $\varphi_n$ 
encodes a winning strategy for player $E$ in the $2^{2^n}$-tiling game. Such a strategy will yield a 
successful tiling of the $2^{2^n}$-corridor for any counterstrategy of player $A$, and any such tiling is 
encoded on a path of $\mathcal{T}_n$ which contains each row of length $2^{2^n}$ as a segment of which 
there are $2^{2^{2^n}}$ many. Thus, $\mathcal{T}_n$ has to have height at least $2^{2^n}\cdot 2^{2^{2^n}}$.
}
{ 
Consider the following instances of the 2EXP-tiling game: $I_n := (T,H,V,F,L,n)$ where $T = \{0,1\}
\times \{\blow,\bflip,\bstay\}$. Tiles are supposed to model bit values in their first component. The second
component describes whether or not a bit has to be \emph{flipped} (value $\bflip$) or remains the
\emph{same} (value $\bstay$) in the increase of a value encoded in binary through these bits. The value
$\blow$ is used to mark the \emph{lowest} bit in a number. Remember that in binary increase the lowest
bit always gets flipped whereas the flipping of any other bit is determined by the value of itself,
the value of the next lower bit and the question whether or not that bit gets flipped. We denote a
tile as $0^\blow$ for instance rather than $(0,\blow)$.

$I_n$ is constructed in a way that forces both players to put down the number $0$ in binary coding
into the first row of the $2^{2^n}$-arena and, whenever a row encodes a number $i$, then the players
need to place the binary encoding of $i+1$ into the next row. Thus, there will be (almost) no choices
for the players. Player $E$ should win when the highest possible number $2^{2^{2^n}}$ is placed in a
row. 

The starting constraints are $F := \{ 0^\blow, 0^\bstay \}$. Hence, the first row
must encode $0$. The horizontal constraints are as follows.
\begin{displaymath}
H \enspace := \enspace \{ (0^\blow,b^\bstay), (1^\blow,b^\bflip), (0^\bflip,b^\bstay), (1^\bflip,b^\bflip), 
(0^\bstay,b^\bstay), (1^\bstay,b^\bstay) \}
\end{displaymath}
where $b$ is, in any case, an arbitrary value in $\{0,1\}$.

Now note that with this $H$ and $F$, there are only two possible first rows that the players can lay down:
$0^\blow 0^\bstay \ldots 0^\bstay$ or $0^\bstay \ldots 0^\bstay$, and it is player $E$ who determines 
entirely through his first choice which of these it is going to be.

Next we will translate the informal description of binary increase given above into the vertical
constraints.
\begin{displaymath}
V \enspace := \enspace \{ (0^\blow,1^\blow), (1^\blow,0^\blow), (0^\bflip,1^x), (1^\bflip,0^x), (0^\bstay,0^x), 
(1^\bstay,1^x) \}
\end{displaymath}
where $x$ is, in any case this time, an arbitrary value in $\{\bflip,\bstay\}$.

Now note that, if a row $i$ is layed down entirely, then the vertical constraints determine uniquely the
bit value of each tile in the next row. Furthermore, if the first tile in row $i$ is of the form 
$b^\blow$ then the first tile in row $i+1$ is uniquely determined to be $(1-b)^\blow$, and $H$ as well
as the bit values in row $i+1$ uniquely determine the flip-values of all the bits in row $i+1$. Furthermore,
all lowest bits that are all $1$ in row $i$ are $0$ in row $i+1$, the lowest bit that is $0$ in row $i$ is
$1$ in row $i+1$, and all other bits in row $i+1$ retain their value from row $i$. Hence, if row $i$
encodes the number $i$ in binary (starting with $0$), then row $i+1$ encodes the number $i+1$ in binary.

Thus, if player $E$ chooses tile $0^\blow$ as the first one, then both players have no choice but to
lay down the binary encodings of $0,1,2,\ldots$. On the other hand, if player $E$ chooses tile $0^\bstay$
as the first one then all rows will encode the number $0$ because the entire arena must be tiled with
$0^\bstay$ only. 

Finally, remember that the goal is to construct $I_n$ in a way that enforces a play filling $2^{2^{2^n}}$ many
rows. This can now easily be achieved by constuction the end constraints in a way that player $E$ only
wins when the number $2^{2^{2^n}}-1$ has been placed down in a row. We therefore set $L := \{ 1^\blow, 1^\bflip \}$.
Note that in the successive increase as constructed above, the last row cannot contain the tile $1^\bstay$. 
}
\qed
\end{proof}


Using the ideas of the transformation mentioned in Theorem \ref{sameExpPower} we can show that the lower bound for \Honeplus is optimal.
Even for strictly more expressive logics than \Honeplus the satisfiability problem remains in \Texptime.
\begin{theorem} \label{complUpperBound}
The satisfiability problem for \Honeplus is \Texptime-complete.
\end{theorem}
\begin{proof}
\mfcsorfull{ 
  The lower bound follows from Thm.~\ref{compl-lower-bound}. The upper bound of \Texptime \ also
  holds when \Honeplus is extended by the fairness operators $\Finfty$ and $\Ginfty$ and the operators
  $\Y$ (previous) and $\Since$ (since) \cite{KupfermanP95} which are the past counterparts of $\X$ and $\U$. 
  The proof is by an exponential reduction to the satisfiability problem of \Hone
  extended by $\Y$ and $\Since$ which is \Dexptime-complete \cite{Weber07}. It should be noted
  that because of Thm.~\ref{theo:fairness} the extension of \Honeplus by $\Finfty$ yields a strictly
  more expressive logic. 
  \qed
}
{ 
Let \HoneplusPastFair be the logic \Honeplus augmented with the operators $\Y$, $\Since$, $\Finfty$ and $\Ginfty$ and similarly \HonePast be \Hone augmented with $\Y$ and $\Since$.  
We describe how to transform a \HoneplusPastFair-formula $\varphi$
into an satisfiability equivalent \HonePast-formula $\varphi'$. The
transformation algorithm we use yields an exponential blowup in the
formula length. As satisfiability of \HonePast is in \Dexptime
\cite{Weber07}  we get the desired \Texptime upper bound for satisfiability of \HoneplusPastFair.
 
The transformation algorithm uses the equivalences already used in the
proof of Theorem \ref{sameExpPower} plus additional equivalences to
deal with the extra operators. For convenience of the reader, we state
all equivalences in the following.  
\begin{enumerate}[(1)]
\item $\neg \X \varphi \equiv \X \neg \varphi$
\item $\neg \Y \varphi \equiv \Y \neg \varphi$
\item $\neg (\varphi \U \varphi') \equiv [(\varphi \wedge \neg \varphi') \U (\neg \varphi \wedge \neg \varphi')] \vee \G \neg \varphi'$
\item $\neg (\varphi \Since \varphi') \equiv (\varphi \wedge \neg \varphi') \Since (\neg \varphi \wedge \neg \varphi')$
\item $\neg \Ginfty \varphi \equiv \Finfty \neg \varphi$
\item $\E(\psi \vee \psi') \equiv \E \psi \vee \E \psi'$
\item $\X\varphi \wedge \X \varphi' \equiv \X(\varphi \wedge \varphi')$
\item $\Y\varphi \wedge \Y \varphi' \equiv \Y(\varphi \wedge \varphi')$
\item $\G\varphi \wedge \G \varphi' \equiv \G(\varphi \wedge \varphi')$
\item $\Ginfty \varphi \wedge \Ginfty \varphi' \equiv \Ginfty(\varphi \wedge \varphi')$
\item Extraction of past operators
\begin{align*}
&\E[\bigwedge_{i=1}^{k} \Y \varphi_i \wedge \bigwedge_{i=1}^{l} (\psi_i \Since \psi'_i) \wedge \X \chi \wedge \G \xi \wedge \Ginfty \rho \wedge
\bigwedge_{i=1}^{m} (\eta_i \U \eta'_i) \wedge \bigwedge_{i=1}^{n} \Finfty \kappa_i \wedge \bigwedge_{i=1}^{o} \neg \Finfty \lambda_i] \\
&\equiv \\
&\bigwedge_{i=1}^{k} \Y \varphi_i \wedge \bigwedge_{i=1}^{l} (\psi_i \Since \psi'_i) \wedge
\E[\X \chi \wedge \G \xi \wedge \Ginfty \rho \wedge \bigwedge_{i=1}^{m} (\eta_i \U \eta'_i) \wedge \bigwedge_{i=1}^{n} \Finfty \kappa_i \wedge \bigwedge_{i=1}^{o} \neg \Finfty \lambda_i]
\end{align*}
\item Elimination of the $\X$-operator
\begin{align*}
&\E[\X \varphi \wedge \G \psi \wedge \Ginfty \chi \wedge \bigwedge_{i=1}^{l} (\xi_i \U \xi'_i) \wedge \bigwedge_{i=1}^{m} \Finfty \kappa_i \wedge \bigwedge_{i=1}^{n} \neg \Finfty \lambda_i] \\
&\equiv \\
&\bigvee_{I \subseteq \{1,...,l\}}[\bigwedge_{i \in I}\xi'_i \wedge \psi \wedge \bigwedge_{i \notin I} \xi_i \wedge \E\X(\varphi \wedge \E[\G \psi \wedge \Ginfty \chi \bigwedge_{i \notin I} (\xi_i \U \xi'_i)\wedge \bigwedge_{i=1}^{m} \Finfty \kappa_i \wedge \bigwedge_{i=1}^{n} \neg \Finfty \lambda_i])]
\end{align*}
\item Disjunction over all possible sequences in which the formulas $\xi'_i$ with $1 \leq i \leq l$ can occur
\begin{align*} 
&\E[\G \psi \wedge \Ginfty \chi \wedge \bigwedge_{i=1}^l (\xi_i \U \xi'_i) \wedge \bigwedge_{i=1}^{m} \Finfty \kappa_i \wedge \bigwedge_{i=1}^{n} \neg \Finfty \lambda_i] \\
&\equiv \\
& \bigvee_{\pi \in Perm(\{1,...,n\})}[E[(\bigwedge_{i=1}^n \xi_i \wedge \psi)\U(\xi'_{\pi(1)} \wedge \E[(\bigwedge_{i \not= \pi(1)}\xi_i \wedge \psi)\U(\xi'_{\pi(2)}\wedge \\
& \ \ \ \ \ \ \ \ \ \ \ \ \ \ \E[(\bigwedge_{i \not= \pi(1),\pi(2)}\xi_i \wedge \psi)\U(\xi'_{\pi(3)}\wedge ...\U(\xi'_{\pi(n)} \wedge \E[\G \psi \wedge \Ginfty \chi \wedge \bigwedge_{i=1}^{m} \Finfty \kappa_i \wedge \bigwedge_{i=1}^{n} \neg \Finfty \lambda_i])...)])])]]
\end{align*}
\item Elimination of the $\Ginfty$-operator
\begin{align*}
\E[\G \psi \wedge \Ginfty \chi \wedge \bigwedge_{i=1}^{m} \Finfty \kappa_i \wedge \bigwedge_{i=1}^{n} \neg \Finfty \lambda_i] 
\equiv   
\E[\psi \U (\E[\G(\psi \wedge \chi) \wedge \bigwedge_{i=1}^{m} \Finfty \kappa_i \wedge \bigwedge_{i=1}^{n} \neg \Finfty \lambda_i])]
\end{align*}
\item Elimination of $\neg \Finfty$
\begin{align*}
\E[\G\psi \wedge \bigwedge_{i=1}^{m} \Finfty \kappa_i \wedge \bigwedge_{i=1}^{n} \neg \Finfty \lambda_i]
\equiv
\E[\psi \U (\E[\G(\psi \wedge \bigwedge_{i=1}^{n} \neg \lambda_i) \wedge \bigwedge_{i=1}^{m} \Finfty \kappa_i])]
\end{align*}
\item Elimination of the $\Finfty$-operator
\begin{align*}
\E[\G\varphi \wedge \bigwedge_{i=1}^{m} \Finfty \kappa_i]
\end{align*}
\centerline{is satisfiable if and only if}
\begin{align*}
\bigwedge_{i=1}^{m} \A\G( \neg \E\G(p_i \wedge \neg \kappa_i)) \wedge \E\G(\varphi \wedge \bigwedge_{i=1}^{m} p_i)
\end{align*}
\centerline{is satisfiable}   
\end{enumerate}
Note that in (16) the two formulas are equivalent only with respect to
satisfiability. For every formula $\kappa_i$, the new formula uses an
additional proposition $p_i$ which is supposed to hold on all paths
satisfying $\Finfty\kappa_i$. 

Let $\varphi$ be a \HoneplusPastFair-formula. We can assume that
$\varphi$ does not contain the path quantifier $\A$  because of
$\A\psi \equiv \neg \E \neg \psi$. In a bottom up fashion the
algorithm replaces each subformula $\E\psi$ of $\varphi$ by a
\HonePast-formula. Thus, it only remains to describe how to transform
a formula $\E\psi$ where $\psi$ is a boolean combination of path
formulas of the form $\Y \chi$, $\chi \Since \chi'$, $\X \chi$, $\chi
\U \chi'$, $\Finfty \chi$ and $\Ginfty \chi$ with $\chi \in
\HonePast$. The transformation involves the following steps:
\begin{itemize}
\item[-] Using De Morgan's laws the $\neg$-operators are pushed to the
  leaves of the Boolean combination.
\item[-] By applying equivalences (1)-(5) negations can be
  eliminated from the outermost Boolean combination and a formula
  $\E\psi$  is obtained in which $\psi$ is a positive Boolean
  combination of path formulas of the form $\Y \chi$, $\chi \Since \chi'$, $\X \chi$, $\G\chi$,$\chi \U \chi'$, $\Finfty \chi$, $\neg \Finfty \chi$ and $\Ginfty \chi$ 
\item By applying equivalences (6)-(10) the formula can be transformed
  into a formula of the form
\newline
$\E[\bigwedge_{i=1}^{k} \Y \varphi_i \wedge \bigwedge_{i=1}^{l} (\psi_i \Since \psi'_i) \wedge \X \chi \wedge \G \xi \wedge \Ginfty \rho \wedge
\bigwedge_{i=1}^{m} (\eta_i \U \eta'_i) \wedge \bigwedge_{i=1}^{n} \Finfty \kappa_i \wedge \bigwedge_{i=1}^{o} \neg \Finfty \lambda_i]$. 
\item Eventually, applying equivalences (11) - (16) a
  \HonePast-formula $\varphi'$ is obtained which is satisfiable if and only if $\varphi$ is satisfiable. 
\end{itemize}
It can be shown that the factorial blowup in equivalence (7) is the worst blowup in the transformation algorithm \cite{EmersonH82}. 
As $n! = 2^{\bigO(n\ log\ n)}$ we can conclude that  $|\varphi'|$ is at most exponential in $|\varphi|$.   
}
\qed
\end{proof}


\section{The Succinctness of \Honeplus w.r.t.\ \Hone} \label{sec:succinctness}
In Corollary \ref{succUpperBound} an upper bound of $2^{\bigO (n \log n)}$ for the succinctness of \Honeplus in \Hone is given.
In this section we establish the lower bound for the succinctness between the two logics.  
Actually we show that \Honeplus is exponentially more succinct than \Hone. The
model-theoretic approach we use in the proof is inspired by \cite{ipl-ctlplus08}. 
We first establish a
kind of small model property for \Hone. 
\begin{theorem}
\label{thm:shallowmodels}
Every finitely satisfiable \Hone-formula $\varphi$ with $\size{\varphi}= n$ has a model of depth $2^{2^{\bigO(n)}}$.
\end{theorem}     
\begin{proof}
In \cite{Weber07} it was shown that for every \Hone-formula $\varphi$, an equivalent
nondeterministic B\"uchi tree automaton $A_{\varphi}$ with $2^{2^{\bigO(\size{\varphi})}}$
states can be constructed. It is easy to see by a pumping argument
that if $A_{\varphi}$ accepts some finite tree at all, it accepts one of depth
$2^{2^{\bigO(\size{\varphi})}}$.
It should be noted that the construction in \cite{Weber07} only
constructs an automaton that is equivalent to $\varphi$ with respect to
satisfiability. However, the only non-equivalent transformation
step is from $\varphi$ to a formula $\varphi'$ without nested
occurrences of the $\downarrow$-operator (Lemma 4.3 in
\cite{Weber07}). It is easy to see that this step only affects the
propositions of models but not their shape let alone depth. 
\qed
\end{proof}

Corollary \ref{cor:deepmodels} and Theorem \ref{thm:shallowmodels}
together immediately yield the following.

\begin{corollary}
\Honeplus is exponentially more succinct than \Hone.
\end{corollary}

\section{Conclusion}

The aim of this paper is to contribute to the understanding of
one-variable hybrid logics on trees, one of the extensions of temporal
logics with reasonable complexity. We showed that \Honeplus has no
additional power over \Hone but is exponentially more succinct, we settled the complexity of
\Honeplus and showed that hybrid variables do not help in expressing
fairness (as \HCTLplus cannot express $\E\G\F p$).

However, we leave a couple of issues for further study, including
the following.
\begin{itemize}
\item We conjecture that the succinctness gap between \Honeplus and
  \Hone is actually $\theta(n)!$.
\item We expect the \HCTL-game to capture exactly the
  expressive power of \HCTL. Remember that here we needed and showed 
  only one part of this equivalence.
\item The complexity of Model Checking for \HCTL has to be explored
  thoroughly, on trees and on arbitrary transition systems. In this
  context, two possible semantics should be explored: the one, where
  variables are bound to nodes of the computation tree and the one
  which binds nodes to states of the transition system (the latter
  semantics makes the satisfiability problem undecidable on arbitrary
  transition systems \cite{ArecesBM01})
\end{itemize}



\bibliographystyle{abbrv}
\bibliography{HBTL}

\ignore{
\input{appendix-expressivity}
\input{appendix-hctl+}
\input{appendix-compl-upper-bound}

\newpage
\appendix
}
\end{document}